\def\confversion{0}
\def\ifconf{\ifnum\confversion=1}
\def\ifnotconf{\ifnum\confversion=0}

\documentclass[11pt]{article}

\usepackage[utf8]{inputenc} %Input encoding
\usepackage[T1]{fontenc} %Output encoding
\usepackage{ae} %Output font fix
\usepackage{aecompl} %Output font fix
\usepackage{xspace} %For \xspace
\usepackage{amsmath} %Advanced math typing
\usepackage{amsthm} %For theorem styles
\usepackage{amsfonts} %For \mathbb
\usepackage{amssymb} %For more math symbols
\usepackage{mathtools} %For extensible math symbols
\usepackage{bm} %For bold symbols with \bm
\usepackage{bbm} %For alternative blackboardsymbols with \mathbbm
\usepackage{thm-restate} %For restatable theorems
\usepackage[inline]{enumitem} %For enumerate styles
\usepackage{array} %For special commands in tables
\usepackage{keyval} %For key=value arguments
\usepackage{xcolor} %For colored text
\usepackage{microtype} %Corrects kerning
\usepackage{empheq} %For empheq environment
\usepackage{tikz} %For tikzfigures
\usetikzlibrary{calc} %For coordinate calculations
\usetikzlibrary{3d} %For 3d figures
\usepackage[in]{fullpage} %Gets rid of stupidly large margins leaving 1 inch on every side
\usepackage{setspace} %For changing spacing between lines anywhere
\usepackage{indentfirst} %Makes sure that the first paragraph is indented
\usepackage{subcaption} % For captions in subfloats
\usepackage[%
  plainpages=false,% Usage of not only arabic numbers
  pdfpagelabels,% Enable page labels
  pagebackref,% Bibliographic back links to pages
  colorlinks,% Colored links
  citecolor=darkgreen,% Citation color
  linkcolor=darkblue,% Link color
  urlcolor=darkred,% URL color
  filecolor=darkred,% File color
  bookmarksopen=false% For starting document with bookmarks tree closed
  pdfstartview=FitH,% For starting document with fit view mode
  pdfpagemode=UseNone,% For starting document with bookmarks tree closed (again)
]{hyperref} %For pdf links
\usepackage[all]{hypcap} %Fixes figures and tables anchors for hyperref
\usepackage[capitalise,nameinlink]{cleveref} %For \cref
\AtBeginEnvironment{quote}{\par\singlespacing\small}
\setlist[enumerate]{label={\roman*.}, ref={(\roman*)}}
\newtheoremstyle{plain}%style name
  {\medskipamount}%space before
  {\smallskipamount}%space after
  {\slshape}%font used
  {0pt}%indentation
  {\bfseries}%modifier theorem head
  {.}%punctuation between theorem head and body
  { }%space after punctuation
  {\thmname{#1}\thmnumber{ #2}{\normalfont\thmnote{ (#3)}}}%theorem specifier

\theoremstyle{plain}

\newtheorem{theorem}{Theorem}[section]
\newtheorem{proposition}[theorem]{Proposition}
\newtheorem{corollary}[theorem]{Corollary}
\newtheorem{lemma}[theorem]{Lemma}

\newtheorem{remark}[theorem]{Remark}
\newtheorem{definition}[theorem]{Definition}

\definecolor{darkred}{rgb}{0.5,0,0}
\definecolor{darkgreen}{rgb}{0,0.35,0}
\definecolor{darkblue}{rgb}{0,0,0.55}
\definecolor{orange}{HTML}{FF7F00}
\def\CC{\ensuremath{\mathbb{C}}}
\def\EE{\ensuremath{\mathbb{E}}}
\def\FF{\ensuremath{\mathbb{F}}}
\def\NN{\ensuremath{\mathbb{N}}}
\def\PP{\ensuremath{\mathbb{P}}}
\def\RR{\ensuremath{\mathbb{R}}}
\def\ZZ{\ensuremath{\mathbb{Z}}}

\newcommand{\F}{{\mathbb F}}
\newcommand{\N}{{\mathbb{N}}}

\def\One{\ensuremath{\mathbbm{1}}}
\def\cF{\ensuremath{\mathcal{F}}}
\DeclareMathOperator{\linspan}{span}
\DeclareMathOperator{\poly}{poly}

\DeclareMathOperator{\supp}{supp}
\DeclareMathOperator{\Lin}{Lin}
\DeclareMathOperator{\GL}{GL}
\DeclareMathOperator{\rk}{rk}
\DeclareMathOperator{\Valid}{Valid}
\DeclareMathOperator{\config}{config}
\DeclareMathOperator{\Config}{Config}
\DeclareMathOperator{\NConfig}{NConfig}
\DeclareMathOperator{\tr}{tr}
\DeclareMathOperator{\MRRW}{MRRW}
\newcommand{\Rdelta}{R_2(\delta)}
\let\emph\textit
\let\geq\geqslant
\let\leq\leqslant
\let\ge\geq

\let\epsilon\varepsilon

\newcommand{\comp}{\mathbin{\circ}}
\newcommand{\rest}{\vert}
\newcommand{\given}[1][]{\mathbin{#1\vert}}
\newcommand{\place}{\mathord{-}}
\newcommand{\pPr}{\widetilde{\PP}}
\def\rn{\bm}
\newcommand{\df}{\stackrel{\text{def}}{=}}
\newcommand{\gQR}[1][\tau]{g_{#1\text{-QR}}}
\newcommand{\GQR}[1][\tau]{G_{#1\text{-QR}}}
\newcommand{\gvu}[1][\tau]{g_{#1\text{-vertex-unif}}}
\newcommand{\Gvu}[1][\tau]{G_{#1\text{-vertex-unif}}}
\def\Lovasz{Lov\'{a}sz\xspace}
\def\Mobius{M\"{o}bius\xspace}
\title{%
  Higher-order Delsarte Dual LPs:\\
  Lifting, Constructions and Completeness%
}

\author{Leonardo Nagami Coregliano\thanks{{\tt University of Chicago}. {\tt lenacore@uchicago.edu}.} \and
        Fernando Granha Jeronimo\thanks{{\tt University of Illinois Urbana-Champaign}. {\tt granha@illinois.edu}.} \and
        Chris Jones\thanks{{\tt Bocconi University}. {\tt chris.jones@unibocconi.it}. } \and
        Nati Linial\thanks{{\tt The Hebrew University of Jerusalem}. {\tt nati@cs.huji.ac.il}. {\tt Supported in part by Advanced ERC Grant PaDiDom 101141253}.} \and
        Elyassaf Loyfer\thanks{{\tt The Hebrew University of Jerusalem}. {\tt elyassaf.loyfer@mail.huji.ac.il}}}

\date{\today}

\begin{document}
\maketitle
\pagenumbering{roman}

\begin{abstract}
A central and longstanding open problem in coding theory is the rate-versus-distance trade-off for binary
error-correcting codes. In a seminal work, Delsarte introduced a family of linear programs establishing
relaxations on the size of optimum codes. To date, the state-of-the-art upper bounds for binary codes
come from dual feasible solutions to these LPs. Still, these bounds are
exponentially far from the best-known existential constructions.

Recently, hierarchies of linear programs extending and strengthening Delsarte's original LPs were introduced for linear
codes, which we refer to as higher-order Delsarte LPs. These new hierarchies were shown to provably converge to the actual
value of optimum codes, namely, they are complete hierarchies. Therefore, understanding them and their dual formulations
becomes a valuable line of investigation. Nonetheless, their higher-order structure poses challenges. In fact, analysis of
all known convex programming hierarchies strengthening Delsarte's original LPs has turned out to be exceedingly difficult
and essentially nothing is known, stalling progress in the area since the 1970s.

Our main result is an analysis of the higher-order Delsarte LPs via their dual formulation.
Although quantitatively, our current analysis only matches the best-known upper bounds, it shows, for the first time,
how to tame the complexity of analyzing a hierarchy strengthening Delsarte's original LPs. In doing so, we reach a better understanding
of the structure of the hierarchy, which may serve as the foundation for
further quantitative improvements. We provide two additional structural results for this hierarchy.
First, we show how to \emph{explicitly} lift any feasible dual solution from level $k$ to a (suitable) larger level
$\ell$ while retaining the objective value. Second, we give a novel proof of completeness using the dual formulation.
\end{abstract}

\thispagestyle{empty}

\newpage
\tableofcontents
\clearpage

\pagenumbering{arabic}
\setcounter{page}{1}

\section{Introduction}

A central and longstanding open problem in coding theory is the rate-vs-distance tradeoff for binary error-correcting codes.
Roughly speaking, it asks for every $\delta \in (0,1/2)$, what is the largest exponent $\Rdelta$ such that there is a distance $\delta n$
error-correcting code of size $2^{\Rdelta\cdot n}$? Despite many decades of effort, the best upper and lower bounds on the rate
$\Rdelta$ are still far apart, implying that we do not understand the exponential growth rate of optimal binary codes.

Convex programming is not only fundamental to algorithm design but it can also be employed to study
combinatorial and mathematical structures. The best known upper bounds on $\Rdelta$ come from the analysis
of convex programming relaxations. In a seminal work, Delsarte~\cite{Del73} showed how to set up linear program
relaxations for the maximum possible size of an error-correcting code. The Delsarte LPs have unfolded into a far-reaching theory leading, for instance, to the best known upper bounds
on $\Rdelta$~\cite{MRRW77}, to breakthroughs in sphere packing~\cite{CE03,V17,CKMRV17}, and to improved bounds on packings
and codes in other types of geometric spaces~\cite{Lev98,Bac06,Bac08,Bar06}.

The success of convex relaxations is sometimes limited by an \emph{integrality gap} between
their optimum  and the true value of the combinatorial problem.
For error-correcting codes, it is known that the value of the Delsarte LP is exponentially far from the Gilbert--Varshamov lower bound~\cite{Sam01}.
If the true size of an optimal binary code is actually near the Gilbert--Varshamov bound (as conjectured by some specialists \cite{JV04,Gop93}),
then this family of relaxations needs to be substantially strengthened.

Given this context, stronger convex relaxations might be imperative to tighten the upper bounds.
In principle, powerful semi-definite programming (SDP) tools such as the Sum-of-Squares hierarchy~\cite{Las15} can be
applied to this problem~\cite{Lau07}. However, asymptotic analysis of these SDP-based relaxations remains
elusive even for the simplest cases~\cite{Sch05}, and only numerical results are known for small constant values of
blocklength~\cite{Gij12}.

To appreciate the difficulty of asymptotically analyzing convex relaxations, recall that the goal is to construct 
a feasible dual solution which upper bounds the primal objective value.
Typically, this requires an explicit construction and analysis.
This is a different goal from typical uses of convex programming in algorithm design, where the starting
point of the analysis is a solution returned by a convex programming solver. There, one does not need
to know the precise structure of the optimum but only the property that it is (near) optimum.

Recently, hierarchies of linear programs extending the Delsarte LPs were proposed for the important case of linear codes~\cite{CJJ22,LL23b}.
We refer to them informally as ``higher-order Delsarte LPs''.
The idea behind them is to strengthen the Delsarte LPs with additional natural
constraints which nonetheless might be simple enough to theoretically analyze. In fact,
these hierarchies were shown to converge to the true size of the code~\cite{CJJ22,CJJ23}, namely, they are complete.
Besides being LPs instead of SDPs, these hierarchies bear strong similarities with Delsarte LPs
for which we now have various theoretical analyses and a richer set of techniques~\cite{MRRW77,FT05,NS05,Bar06,Bar08,NS09,Sam23b,LL23a,CD24}.

Constructing dual solutions for the higher-order Delsarte LPs can lead to 
a breakthrough in the rate-versus-distance problem. Nonetheless, the higher-order
structure of these LPs may still require substantial effort to be understood and analyzed.
In this work, our main goal is to substantially increase our understanding of the structure of the higher-order Delsarte LP hierarchies
by establishing three new results about their dual formulations.

Before we present our results, we first recall these LPs with an informal and intuitive description
(see~\cref{sec:hierarchies,sec:otherform} for more details).
The Delsarte LP (used in the first LP bound) has a variable intended to count the number of codewords of each Hamming weight. The higher-order Delsarte
LPs form a hierarchy with a level parameter $\ell \in \N$. There is a variable intended to count the number of $\ell$-tuples of codewords with every
possible Hamming weight configuration of a subspace of dimension $\ell$.
For example, for $\ell = 2$, essentially there is a variable for each $(a,b,c) \in \{0,1,\dots, n\}^3$ 
which is intended to be the number of pairs of codewords $(x,y)$ such that $(|x|, |y|, |x+y|) = (a,b,c)$.

\subsection{Our Contributions}

We show three different ways of constructing dual solutions for the higher-order Delsarte LPs.
First, we show how to lift a solution from any level $k$ to a higher level $\ell$. 
Second, we show how to construct an explicit solution at a higher level. In contrast with the
lift that takes any solution as a black box, here we must directly understand and tackle the
additional complicated structure imposed by the higher levels. Lastly, by relaxing the constraints,
we are able to come up with a dual solution that shows completeness. We will now elaborate on each
of these three new constructions of higher-order dual solutions.

Motivated by the proven strength of these new hierarchies (their completeness) and
our extensive understanding of the first level of the hierarchy (i.e., Delsarte's original LPs),
a natural question is how to \emph{lift} a dual solution from level $1$ to an arbitrary level $\ell$, i.e., how to explicitly construct a level $\ell$ dual
solution from a level $1$ dual solution while (appropriately) retaining its objective value.
A lift is one way to identify an explicit solution to level $\ell$ of the hierarchy
whose value matches the Delsarte LP.
Therefore, there may be potential to perturb the lifted solution in a direction which improves the objective value.
Besides improving our understanding of how dual solutions are related
to each other across multiple levels of the hierarchy, the additional structure of the dual at higher levels has the potential of leading to improvements in
the objective value (in case the original Delsarte LPs suffer from integrality gap).
We prove a general lifting result from a level $k$
dual solution to level $\ell$ assuming that $k$ divides $\ell$. More precisely, our first structural result is given below.

\begin{theorem}[Lifting Dual Solutions (Informal version of~\cref{thm:lift})]
  Given an arbitrary dual feasible solution of level $k$, we can \emph{explicitly} construct a new dual feasible solution of
  level $\ell \ge k$ provided $k$ divides $\ell$ (this can be done over any finite field $\FF_q$).
  Furthermore, this new dual solution has (appropriately) the same objective value of the given starting solution.
\end{theorem}

\begin{remark}
  Unlike more structured convex programming hierarchies such as the Sum-of-Squares SDP hierarchy or Sherali-Adams LP
  hierarchy, establishing a lift for the higher-order Delsarte dual LPs is not trivial. We also stress that the value of the above theorem
  lies in its \emph{explicitness}; ``monotonicity''
  of the objective value was already established~\cite{CJJ22} (using the primal formulation), and this is not the point of the preceding theorem.  
\end{remark}

Another natural question is whether we can construct dual feasible solutions for higher levels of
these new hierarchies from scratch. As noted above, there are now a wealth of perspectives and techniques
to construct dual feasible solutions to level $1$ (the original Delsarte LPs). For instance, the original MRRW
proof relies on properties of the Krawtchouk polynomials, which form a family of orthogonal
polynomials, whereas some more recent proofs use spectral graph theory and Fourier analysis. Curiously, these various analyses
are largely different perspectives or small variations of a single construction. Nonetheless, having multiple perspectives can be
very helpful, and they can serve as (seemingly) different starting points for analyzing the hierarchies.

Although these hierarchies are structurally similar to the original LPs (coinciding at level $1$), there are challenges
to be addressed. First, the hierarchy at level $\ell \ge 2$ inherently relies on multivariate
versions of Krawtchouk polynomials, as opposed to the univariate version of level $1$. The asymptotic behavior of the first root of univariate Krawtchouk polynomials
plays a crucial role in the original analysis, while establishing an analogous property in the multivariate case is less clear.
Moreover, while level $1$ is the same regardless of whether a code is linear or not (only the meaning of the variables changes), higher
levels of these hierarchies have new constraints associated with linearity which
pose new challenges.

Our second structural and main result is an explicit construction of dual feasible solutions to constant levels of the hierarchy
for the important class of balanced linear codes\footnote{Recall that, for $\epsilon \in (0,1)$, an
$\epsilon$-balanced linear code is a code in which every non-zero codeword has Hamming weight in $[(1-\epsilon)n/2,(1+\epsilon)n/2]$.},
giving the first theoretical analysis of a convex programming hierarchy containing Delsarte's original LP.
The main contribution here is to make sense of the higher-order structure of the hierarchy, suitably generalizing spectral-based techniques
for the Delsarte LP. Obtaining such suitable generalization was met with substantial challenges (see~\cref{sec:overview})
as it may be expected in analyzing \emph{any} convex programming hierarchy strengthening Delsarte's LP since progress in this area has stalled in 1970s.
The objective value of our constructed solutions approximately matches the state-of-the-art MRRW bound up to lower-order terms in $\epsilon$. Our main result is stated below.

\begin{theorem}[Higher-order Dual Solution (Informal version of~\cref{cor:spectral_construction} of~\cref{theo:spectral_construction})]
  For every constant level $\ell\in\NN_+$, there is an explicit construction of dual feasible solutions at level $\ell$ for binary
  $\epsilon$-balanced linear codes with rate upper bound $R_2^{\ell}(\delta)$, with $\delta=(1-\epsilon)/2$, satisfying
  $$
  R_2^{\ell}(\delta) = (1 + o_{\epsilon}(1)) \cdot R_2^{\textup{MRRW}}(\delta),
  $$
  where $R_2^{\textup{MRRW}}(\delta)$ is the rate upper bound of the first LP bound of~\cite{MRRW77}.
\end{theorem}

The proof of the above theorem establishes a footprint of how to construct higher-order dual solutions, breaking the ice on the daunting
complexity of higher-order convex programs.
It may serve as a technical foundation for further quantitative
improvements.

% \begin{remark}
%   We stress that the main contribution of the above theorem is in how it was proven, namely, by taming, for the first time, the daunting complexity of
%   analyzing a convex programming hierarchy, strengthening the original Delsarte's LP for this problem. In doing so, a much better understanding of the structure of the higher-order
%   Delsarte LP hierarchy is reached, which may serve as the foundation for further quantitative improvements.
% \end{remark}

We now give some additional context before describing our third structural result. A feasible solution of the dual can be seen as a certificate
establishing a universal upper bound on the size of codes. Ideally, the better we understand the structure and nature of these dual certificates,
the better positioned we may be for designing new ones. The higher-order Delsarte hierarchies are known to converge to the
true value of a linear code; however, the known proofs \cite{CJJ22, CJJ23} are entirely based on the primal version of these hierarchies. It is then natural to ask
if we can use the dual hierarchies to prove completeness. Our third result is a novel completeness proof
of these hierarchies which uses their dual formulations. 

\begin{theorem}[Completeness from the Dual (Informal version of~\cref{thm:completeness_from_dual})]
  The dual higher-order Delsarte LPs obtain the true value of a linear code for any level $\ell \ge n$ and over any finite field $\FF_q$.
\end{theorem}

\begin{remark}
  Unlike other more structured convex programming hierarchies, such as the Sum-of-Squares SDP hierarchy or Sherali-Adams LP
  hierarchy, (exact) completeness for the higher-order Delsarte's LP is not immediate~\cite{CJJ22,CJJ23}.
\end{remark}

A better understanding of completeness from the dual may also help understand the power of natural LP hierarchies for lattice packings,
extending the celebrated Cohn and Elkies LP for sphere packing~\cite{CE03,V17,CKMRV17}. Recall that the Cohn and Elkies LP can be
seen as a close analog of Delsarte's \emph{dual} LP designed for sphere packing.

\subsection{Organization}

First, we recall the higher-order Delsarte LP hierarchies of~\cite{CJJ22,LL23b}
in~\cref{sec:hierarchies}. We provide several different formulations of the hierarchies which will be used to establish our results (other equivalent formulations that will
not be used in the present work are included in~\cref{sec:otherform} for the curious
reader). In~\cref{sec:overview}, we give the main technical intuition of the proofs. We formally prove the lifting in~\cref{sec:lifting}. The completeness from dual is
presented in~\cref{sec:completeness}. The spectral-based construction of higher-order dual
feasible solutions is given in~\cref{sec:spectral_construction}. We end with some concluding
remarks in~\cref{sec:conclusion}.

The reader should refer to~\cref{sec:notation} for notation as needed.

\section{A Brief Introduction to the Hierarchies}
\label{sec:hierarchies}

Both hierarchies of~\cite{CJJ22,LL23b} can be used to upper bound sizes of linear codes in an
arbitrary set of ``valid'' linear codes $\Valid_n\subseteq L_{\FF_q}(\FF_q^n)$. In the prototypical
cases, $\Valid_n$ is the set of all linear codes of distance at least $d$, or the set of all
$\epsilon$-balanced codes. Once $\Valid_n$ is fixed, at level $\ell\in\NN_+$ the hierarchies make
use of the set
\begin{align*}
  \Valid_{n,\ell} & \df \{X\in\FF_q^{\ell\times n} \mid \linspan(\{X_1,\ldots,X_\ell\})\in\Valid_n\}.
\end{align*}

The easiest way of stating the hierarchy of~\cite{CJJ22} at level $\ell$ is as the
\Lovasz\ $\vartheta'$ of the graph $G_{n,\ell}$ over the vertex set $\FF_q^{\ell\times n}$ in which
$X,Y\in\FF_q^{\ell\times n}$ are adjacent exactly when $X-Y\notin\Valid_{n,\ell}$. If
$C\in\Valid_n$, then the set $\{X\in\FF_q^{\ell\times n} \mid X_1,\ldots,X_\ell\in C\}$ is an
independent set in $G_{n,\ell}$ of size exactly $\lvert C\rvert^\ell$, which is upper bounded by
$\vartheta'(G_{n,\ell})$, giving us the first formulation of the hierarchy of~\eqref{eq:SDPprimal}
(which is deferred to \cref{subsec:Lovasz} as it will not be used in the present paper).

It turns out that the SDP arising in the \Lovasz\ $\vartheta'$ function can be explicitly
diagonalized, leading to a linear program. By noting that there is a natural ``global translation''
action of $\FF_q^n$ on the space $\FF_q^{\ell\times n}$ given by
\begin{align*}
  (z\cdot X)_{jk} & \df X_{jk} + z_k & (X\in\FF_q^{\ell\times n}, z\in\FF_q^n, j\in[\ell], k\in[n]),
\end{align*}
and that the program~\eqref{eq:SDPprimal} of $\vartheta'(G_{n,\ell})$ is $\FF_q^n$-symmetric, every
feasible solution can be symmetrized under this action without violating its feasibility or changing
its value. Furthermore, $\FF_q^n$-symmetric solutions are simultaneously diagonalizable and the
positive semidefinite constraint is then encoded by the Fourier transform (see~\cref{subsec:LP} for
more details) given by
\begin{gather*}
  \widehat{f}(X)
  \df
  \langle f,\chi_X\rangle
  =
  \frac{1}{q^{n\ell}}\sum_{X\in\FF_q^{\ell\times n}} f(X)\overline{\chi_Z(X)}
  \qquad (f\in\CC^{\FF_q^{\ell\times n}}, X\in\FF_q^{\ell\times n}),
  \\
  \chi_Z(X)
  \df
  \exp\left(\sum_{j\in[\ell]}\sum_{k\in[n]}\frac{2\pi i X_{jk} Z_{jk}}{q}\right)
  \qquad (X\in\FF_q^{\ell\times n}).
\end{gather*}

This yields the linear program~\eqref{eq:LPprimal} below, whose dual is~\eqref{eq:LPdual} and that
first appeared in~\cite{CJJ22}. A linear code $C\in\Valid_n$ yields a natural solution $f_C$
of~\eqref{eq:LPprimal} given by $f_C(X)\df \One[X_1,\ldots,X_\ell\in C]$, whose value is $\lvert
C\rvert^\ell$. Note that when $q$ is a power of $2$, due to $X = -X$, the symmetry constraints in
the primal are automatically enforced and we can therefore remove $\beta$ from the dual.

\begin{empheq}[box=\fbox]{equation}\label{eq:LPprimal}
  \begin{aligned}
    \text{Variables: }
    & \mathrlap{f\colon \FF_q^{\ell\times n}\to\RR}
    \\
    \max \qquad
    & \sum_{X \in \FF_q^{\ell\times n}} f(X)
    \\
    \text{s.t.} \qquad
    & f(0) = 1
    & &
    & (\text{Normalization})
    \\
    & f(X) = 0
    & & \forall X\in\FF_q^{\ell\times n}\setminus\Valid_{n,\ell}
    & (\text{Validity})
    \\
    & \widehat{f}(X) \geq 0
    & & \forall X \in \FF_q^{\ell\times n}
    & (\text{Fourier})
    \\
    & f(X) \geq 0
    & & \forall X \in \FF_q^{\ell\times n}
    & (\text{Non-negativity})
    \\
    & f(X) = f(-X)
    & & \forall X \in \FF_q^{\ell\times n}
    & (\text{Symmetry})
  \end{aligned}
\end{empheq}

\begin{empheq}[box=\fbox]{equation}\label{eq:LPdual}
  \begin{aligned}
    \text{Variables: }
    & \mathrlap{g\colon\FF_q^{\ell\times n} \to \RR, \beta\colon\FF_q^{\ell\times n}\to\RR}
    \\
    \min \qquad
    & g(0)
    \\
    \text{s.t.} \qquad
    & \widehat{g}(0) = 1
    & &
    & (\text{Normalization})
    \\
    & g(X) + \beta(X) - \beta(-X)\leq 0
    & & \forall X\in\Valid_{n,\ell}\setminus\{0\}
    & (\text{Validity})
    \\
    & \widehat{g}(X) \geq 0
    & & \forall X \in \FF_q^{\ell\times n}
    & (\text{Non-negativity})
  \end{aligned}
\end{empheq}

Next, one observes that there is a natural ``label permutation'' action of $S_n$ on
$\FF_q^{\ell\times n}$ given by
\begin{align*}
  (\sigma\cdot X)_{ij} & \df X_{i\sigma(j)}
  \qquad (X\in\FF_q^{\ell\times n}, \sigma\in S_n, i\in[\ell], j\in[n]).
\end{align*}
It is easy to see that if $\Valid_n$ is $S_n$-symmetric under the natural action of $S_n$ on
$\FF_q^n$, then so are $\Valid_{n,\ell}$ and~\eqref{eq:LPprimal} under the $S_n$-action above. This
allows us to further symmetrize the program to obtain the formulation in~\eqref{eq:KLPprimal} in
which the Fourier transform is encoded using multivariate Krawtchouk polynomials
(see~\cref{subsec:kraw}).

\medskip

Finally, we introduce the Partial Fourier Hierarchy of~\cite{LL23b}. This hierarchy follows from the
observation that the natural solutions $f_C(X)\df\One[X_1,\ldots,X_\ell\in C]$
to~\eqref{eq:LPprimal} not only have non-negative Fourier transforms, but in fact have non-negative
``partial Fourier transforms'' defined as follows.

First, we note that $\GL_\ell(\FF_q)$ also acts on $\FF_q^{\ell\times n}$ by left-multiplication,
which in turn induces a right-action of $\GL_\ell(\FF_q)$ on the set of functions $\FF_q^{\ell\times
  n}\to\CC$ given by $(f\cdot M)(X)\df f(M\cdot X)$. Then for $X,Y\in\FF_q^{\ell\times n}$,
$k\in\{0,1,\ldots,n\}$ and $M\in\GL_\ell(\FF_q)$, we let
\begin{align*}
  \chi_Y^{(k)}(X)
  & \df
  q^{(\ell-k)n}
  \cdot\left(\prod_{j=1}^k \chi_{Y_j}(X_j)\right)
  \cdot\left(\prod_{j=k+1}^n\One_{Y_j}(X_j)\right),
  &
  \chi_Y^{k,M}(X) & \df \chi_{M^{-1}\cdot Y}^{(k)}(M^{-1}\cdot X),
\end{align*}
where $\chi_y(x)\df\exp(\sum_{j\in[n]}2\pi i y_j x_j/q)$ is the usual character and we let
\begin{align*}
  \cF_k(f)(X)
  & \df
  \langle f,\chi_X^{(k)}\rangle
  =
  \frac{1}{q^{\ell n}}\cdot\sum_{Z\in\FF_q^{\ell\times n}} f(Z)\cdot\overline{\chi}_X^{(k)}(Z),
  & 
  \cF_{k,M}(f)(X)
  & \df
  \langle f,\chi_X^{k,M}\rangle,
\end{align*}
for every $f\colon\FF_q^{\ell\times n}\to\CC$. A straightforward calculation then yields
\begin{align}\label{eq:cFkMcFk}
  \cF_{k,M}(f) & = \cF_k(f\cdot M)\cdot M^{-1},
  &
  \cF_{k,M}^{-1}(f)
  & =
  q^{kn}\cdot\cF_{k,M}(f)\cdot R_k,
\end{align}
where $R_k$ is the diagonal matrix whose diagonal consists of $k$ entries $-1$ followed by $\ell-k$
entries $1$.

Noting that for every $C\in L_{\FF_q}(\FF_q^n)$ the function $f_C(X)\df\One[X_1,\ldots,X_\ell\in C]$
satisfies $\cF_{k,M}(f_C)\geq 0$ ($k\in[\ell]$, $M\in\GL_\ell(\FF_q)$), it follows that we can add
further constraints to~\eqref{eq:LPprimal} to obtain a stronger hierarchy,\footnote{In fact,
  \cite{LL23b} only includes partial Fouriers with $M=I$, but explicitly requires solutions to be
  $\GL_\ell(\FF_q)$-symmetric; here we opt for this formulation which can be shown to be equivalent
  straightforwardly.} called the partial Fourier hierarchy~\cite{LL23b}, formulated
in~\eqref{eq:pLPprimal} and whose rather technical dual~\eqref{eq:pLPdual} is deferred to
\cref{sec:lifting}. We will show in \cref{lem:symmpLPdual} that the dual of~\eqref{eq:pLPprimal} is
further equivalent to the simpler~\eqref{eq:symmpLPdual} below.

\begin{empheq}[box=\fbox]{gather}\label{eq:pLPprimal}
  \begin{aligned}
    \text{Variables: }
    & \mathrlap{f\colon\FF_q^{\ell\times n}\to\RR}
    \\
    \max \qquad
    & \sum_{X \in \FF_q^{\ell\times n}} f(X)
    \\
    \text{s.t.} \qquad
    & f(0) = 1
    & &
    & (\text{Normalization})
    \\
    & f(X) = 0
    & & \forall X\in\FF_q^{\ell\times n}\setminus\Valid_{n,\ell}
    & (\text{Validity})
    \\
    & \cF_{k,M}(f)(X) \geq 0
    & & \forall X \in \FF_q^{\ell\times n}, \forall k\in[\ell], \forall M\in\GL_\ell(\FF_q)
    & (\text{Partial Fourier})
    \\
    & f(X) \geq 0
    & & \forall X \in \FF_q^{\ell\times n}
    & (\text{Non-negativity})
    \\
    & f(X) = f(-X)
    & & \forall X \in \FF_q^{\ell\times n}
    & (\text{Symmetry})
  \end{aligned}
\end{empheq}

\begin{empheq}[box=\fbox]{gather}\label{eq:symmpLPdual}
  \begin{aligned}
    \text{Variables: }
    & \mathrlap{g_k\colon\FF_q^{\ell\times n}\to\RR \; (k\in[\ell])}
    \\
    \min \qquad
    & 1 + \sum_{k\in[\ell]} g_k(0)
    \\
    \text{s.t. } \qquad
    & 1
    + \frac{1}{\lvert\GL_\ell(\FF_q)\rvert}\cdot
    \sum_{\substack{k\in[\ell]\\M\in\GL_\ell(\FF_q)}} (g_k\cdot M)(X)
    \leq 0
    & & \forall X\in\Valid_{n,\ell}\setminus\{0\}
    & (\text{Validity})
    \\
    & \cF_k(g_k)\geq 0
    & & \forall k\in[\ell]
    & \mathllap{(\text{Partial Fourier})}
  \end{aligned}
\end{empheq}

\section{Technical Overview of the Proofs}\label{sec:overview}

The purpose of this section is to highlight the main ideas of the proofs
and provide intuition, in preparation for the full results.
For simplicity, we restrict ourselves here to $q=2$.

\subsection{Lifting Dual Solutions}

A lift transforms a level-$k$ solution of value $V$ into a level-$\ell$ solution, with objective
value $V^{\ell/k}$.  The scaling is correct, since solutions in the hierarchy's $\ell$-th level
provide an upper-bound on $\lvert C\rvert^\ell$ for $C\in\Valid_n$. In our approach we construct
functions $f^{(1)},f^{(2)},\ldots, f^{(\ell/k)}$ that satisfy increasingly more constraints, and
terminate with a feasible solution $f^{(\ell/k)}$.

Here we illustrate our method with the LP~\eqref{eq:LPdual} over $\FF_2$.  In~\cref{sec:lifting},
the lifts are developed in full for the stronger LP~\eqref{eq:symmpLPdual} over general finite
fields, $\FF_q$.

We start with a lift from level $1$. Let $h'\colon\F_2^n\to\RR$ be a feasible solution
for level 1 of the dual hierarchy~\eqref{eq:LPdual}. It will be more convenient to work with $h \df
h'-1$. Observe that
\begin{align*}
    \widehat{h} & \geq 0,
    &
    \widehat{h}(0) & = 0,
    &
    \forall x\in\Valid_{n,1}\setminus\{0\}, h(x) & \leq -1,
\end{align*}
To lift $h$ to level $\ell$ we start by defining
\begin{align*}
  f^{(1)}(X) & \df h(X_1)
  \qquad
  \forall X=(X_1,\ldots,X_\ell) \in \FF_2^{\ell\times n}.
\end{align*}
Namely, we ignore all of the rows of $X$ except for the first.

The Fourier transform of $f^{(1)}$ is non-negative, since $\widehat{h}\geq 0$. Also,
$\widehat{f}^{(1)}(0)=0$. These two properties persist throughout the process, for
$f^{(2)},f^{(3)},\ldots$ etc.

The validity constraints are only satisfied if $X_1\in\Valid_{n,1}\setminus\{0\}$: otherwise,
$f^{(1)}(X) = h(0)$, which not only is positive but in fact exponentially large. To handle the case
$X\in\Valid_{n,\ell}$ with $X_1=0$, we define
\begin{align*}
    f^{(2)}(X) 
    & \df
    f^{(1)}(X)
    + (1+h(0)) 
    \cdot h(X_2)
    \cdot \One[X_1=0]
\end{align*}
Observe that $f^{(2)}$ only differs from $f^{(1)}$ when $X_1=0$. The validity constraints now hold
if $(X_1,X_2)\in\Valid_{n,2}$, but not when $X_1=X_2=0$.

We continue by defining, for $t=3,\ldots,\ell$,
\begin{align*}
    f^{(t)}(X) 
    & \df
    f^{(t-1)}(X) 
    + (1+h(0))^{t-1}
    \cdot h(X_t) 
    \cdot \One[X_{1,\ldots,t-1}=0]
\end{align*}
It is not hard to verify that $f^{(t)}$ satisfies\footnote{In fact, there holds, moreover:
$f^{(t)}(X) \leq -1$ for every $X$ such that $X_{1,\ldots,t}\neq 0$ and
$\{X_1,\ldots,X_t\}\subset \Valid_{n,1}$.}
\begin{gather*}
  \begin{aligned}
    \widehat{f}^{(t)} & \geq 0,
    \qquad & \qquad
    \widehat{f}^{(t)}(0) & = 0,
  \end{aligned}
  \\
  \begin{aligned}
    f^{(t)}(X) & \leq -1
    & & \forall X\in\FF_q^{\ell\times n} \text{ with } X_{1,\ldots,t}\in \Valid_{n,t} \setminus\{0\}
    \\
    f^{(t)}(X) & = (1+h(0))^{t}-1
    & & \forall X\in\FF_q^{\ell\times n} \text{ with } X_{1,\ldots,t}= 0.
  \end{aligned}
\end{gather*}
Thus, the function $f\df f^{(\ell)}+1$ is a feasible solution and $f(0) = (1 + h(0))^\ell = h'(0)^\ell$. This
concludes the lift from level $1$ to level $\ell$ in the LP~\eqref{eq:LPdual}.

The lift from level $k$ to level $\ell$ proceeds similarly, except that we advance in chunks of $k$
rows per step. Suppose we have a level-$k$ feasible solution to the
hierarchy~\eqref{eq:LPdual}, $h'\colon\F_2^{k\times n}\to\RR$, and let $h\df h'-1$. We define
\begin{align*}
    f^{(0)} & \df 0
    \\
    f^{(t)} 
    & \df
    f^{(t-1)} 
    + (1+h(0))^{t-1} 
    \cdot h(X_{k\cdot (t-1)+1,\ldots,k\cdot t})
    \cdot \One[X_{1,\ldots,k\cdot(t-1)}]
    \qquad (t\in[\ell/k]).
\end{align*}
Then, similar arguments show that $f\df f^{(\ell/k)} + 1$ is feasible for level $\ell$, and its
value is $f(0) = h'(0)^{\ell/k}$.

Our strategy remains unchanged
as we move to the stronger
LP \eqref{eq:symmpLPdual}. 
However, the
symmetry operation in the validity constraints calls for a slight change 
in the argument. Rather than arguing in
terms of the number of zero rows in $X$,
we now account by $X$'s rank.
Let $h_1,\ldots,h_k$ be a feasible solution to level $k$, that is
\begin{align*}
  \cF_{i}(h_i) & \geq 0
  & &
  \forall i\in[k],
  \\
  1 + \sum_{i=1}^{k} 
  \EE_{\rn{M}\sim U(\GL_\ell(\FF_q))}[
    h_i(M\cdot X)]
  & \leq 0
  & &
  \forall X\in\Valid_{n,k} \setminus \{0\},
\end{align*}
where $U(\GL_\ell(\FF_q))$ is the uniform distribution on $\GL_\ell(\FF_q)$ and the value of
$(h_1,\ldots,h_k)$ is $V_h\df 1 + \sum_{i\in[k]} h_i(0)$.

We would like to put the information of this level-$k$ solution in the top $k$ levels of a
level-$\ell$ solution $g$, that is, we would like to put the information $h_1,\ldots,h_k$ into
$g_{\ell-k+1},\ldots,g_\ell$, respectively; we will then set $g_1\df\cdots\df g_{\ell-k}\df
0$. Furthermore, this needs to be organized so that the constraints $\cF_i(g_i)\geq 0$ follow
directly from the constraints $\cF_i(h_i)\geq 0$. To do so, the solution is slightly permuted around
when compared to the previous cases.

For each $i\in[k]$, we define a sequence of functions $f_i^{(1)},f_i^{(2)},\ldots,f_i^{(\ell/k)}$
as follows:
\begin{align*}
    f_i^{(0)} & \df 0,
    &
    f_i^{(t)} 
    & \df
    f_i^{(t-1)} 
    + V_h^{t-1}
    \cdot h_i(X_{\ell-k+1,\ldots,\ell})
    \cdot \One[X_{1,\ldots,kt}=0]
    \qquad (t\in[\ell/k]).
\end{align*}
We will then argue that for every $t\in[\ell/k]$, we have
\begin{align*}
  \cF_{\ell-k+i}(f_i^{(t)}) & \geq 0
  & &
  \forall i\in[k],
  \\
  1 + \sum_{i=1}^{k} 
  \EE_{\rn{M}\sim U(\GL_\ell(\FF_q))}[
    f_i^{(t)}(M\cdot X)]
  & \leq 0
  & &
  \forall X\in\Valid_{n,\ell} \setminus \{0\}
  \text{ with }
  \rk(X)\geq\ell-t\cdot k.
\end{align*}

Consequently, letting $g_i\df 0$ for eveery $i\in[\ell-k]$ and $g_i \df f^{(\ell/k)}_{i-\ell+k}$ for
$\ell-k+1\leq i\leq\ell$, we obtain a feasible solution whose value is $V_h^{\ell/k}$.

\subsection{Spectral-based Construction of Dual Solutions}
\label{subsec:overview:spectral}

We describe now how we use spectral techniques to construct dual solutions for the hierarchy.
We start with an abstract description of the idea, then move to its realization
in Delsarte's case, and finally to the way that we implement it in higher levels of the hierarchy.
Throughout this section, we refer only to the LP hierarchy~\eqref{eq:LPdual} over $\FF_2$,
of which level $1$ is Delsarte's LP.

We begin with the abstract construction, which is mostly inspired by~\cite{NS05}, and presented with
more detail in~\cite{LL22,Sam23a}. Although this abstract construction is relatively intuitive, fully
implementing it for the higher-order hierarchy is far from trivial as the reader will see in this paper.
One can form a feasible $f\colon\FF_2^{\ell\times n}\to\RR$ by defining
\begin{align*}
  f(X) 
  & \df
  \frac{\phi(X)\cdot\Gamma^2(X)}{\widehat{\phi\cdot\Gamma^2}(0)},
\end{align*}
where $\phi,\Gamma\colon\FF_2^{\ell\times n}\to\RR$ are not identically zero, and satisfy
\begin{align*}
  \forall X\in\Valid_{n,\ell}\setminus\{0\}, \phi(X) & \leq 0,
  &
  \widehat{\Gamma} & \geq 0,
  &
  2^{n\ell} \widehat{\phi} * \widehat{\Gamma} & \geq \widehat{\Gamma}.
\end{align*}
The sign of $f$ is governed by $\phi$, hence it fulfills the validity constraints. The Fourier
constraints are met, because, by the convolution theorem, $\widehat{f}$ is up to positive constants
equal to $\widehat{\phi}*\widehat{\Gamma}*\widehat{\Gamma}\geq 2^{-n\ell}\widehat{\Gamma}*\widehat{\Gamma}\geq
0$. An upper bound on the objective function
is derived using Cauchy-Schwarz
as follows
\begin{align*}
  f(0)
  \leq 
  2^{n\ell}\phi(0)\cdot \frac{\Gamma^2(0)}{(\widehat{\Gamma}*\widehat{\Gamma})(0)}
  =
  2^{n\ell}\phi(0)\cdot
  \frac{\lVert\widehat{\Gamma}\rVert_1^2}{\lVert\widehat{\Gamma}\rVert_2^2}
  & \stackrel{\text{C.S}}{\leq}
  \phi(0)\cdot\lvert\supp(\widehat{\Gamma})\rvert.
\end{align*}
One usually fixes $\phi$ as a low-degree polynomial, and seeks
a feasible $\Gamma$ so that $\lvert \supp(\widehat{\Gamma}) \rvert$ is minimal.

The operator ``$2^{n\ell}\widehat{\phi}*\place$'' of convolution by $\widehat{\phi}$ (up to
renormalization) can be represented by a matrix which we denote $M_\phi\in\RR^{\FF_2^{\ell\times
    n}\times\FF_2^{\ell\times n}}$, that is, we have
\begin{align*}
    M_\phi h & = 2^{n\ell}\widehat{\phi} * h,
    \quad
    (h\colon\F_2^{\ell\times n}\to\RR).
\end{align*}
Thus, finding $\Gamma$ becomes a spectral problem. 

When $\ell = 1$ and working with distance-$d$ codes, $\phi$ can be as simple as the linear function
$\phi_{\MRRW}(x) = 2(d-\lvert x\rvert)$, which is usually the case. The corresponding matrix is
$M_{\phi_{\MRRW}} = A - (n-2d)I$, where $A$ is the adjacency matrix of the Hamming graph,
$A(x,y)=\One[\lvert x-y\rvert=1]$.  The problem of finding an appropriate $\Gamma$ is well explored.
It can be done through different techniques, e.g., specific properties of Krawtchouk
Polynomials~\cite{MRRW77}, Perron-Frobenius Theorem~\cite{Bar06}, or by taking advantage of the fact
that the matrix $A$ is highly symmetric~\cite{NS05}.

As $\ell$ grows, however, the set $\Valid_{n,\ell}$ becomes increasingly complicated and cannot be
captured or closely approximated by a linear function. Constructing a satisfactory $\phi$ is a
problem in itself, which was first addressed in~\cite{LL22}. Finding a complementary $\Gamma$ was
left by the authors of~\cite{LL22} as an open problem. The methods used to find $\Gamma$ in
Delsarte's case $\ell=1$ are inapplicable here due to the high-dimensionality of the problem, and
the complicated structure of the corresponding matrix $M_\phi$. 

In the current work we solve this open problem for a variation of the suggested polynomial $\phi$,
which is valid for $\epsilon$-balanced linear codes.  The polynomial, denoted $\Phi_m$, is
defined in~\eqref{eq:Phim} and its necessary properties are established in \cref{lem:constraints}.
We find an appropriate $\Gamma$ for $\Phi_m$ which leads to a feasible solution whose value is
equivalent to MRRW, up to lower order terms.

Our strategy is as follows.  First, we show (\cref{lem:Mm}) that the matrix $M_m$, which corresponds
to the operator ``$2^{n\ell}\widehat{\Phi}_m*\place$'', is a sum of terms of the form
\begin{align*}
  \left(\prod_{u\in U}\sum_{\substack{v\in\FF_2^\ell\\\langle u,v\rangle=1}} A_v^m\right)
  \cdot
  \left(
  \frac{1}{k}\cdot\sum_{\substack{v\in\FF_2^\ell\\\langle i,v\rangle=1}} A_v^m
  - \frac{2^{\ell-1}(\epsilon n)^m}{2^\ell-k}
  \right)
\end{align*}
with non-negative coefficients. Here, 
\begin{itemize}
\item $A_v$, for every $v\in\FF_2^\ell\setminus\{0\}$, is the adjacency matrix of a graph over the
  vertex set $\FF_2^{\ell\times n}$, where $X,Y$ are adjacent if $X$ is obtained from $Y$ by adding
  $v$ to one of its columns.
\item $m\in \N$ is even,
\item $1 \leq k \leq 2^{\ell}-1$,
\item $i\in\FF_2^\ell\setminus\{0\}$,
\item $U\subseteq\FF_2^\ell\setminus\{0\}$.
\end{itemize}

Noting that for every $i\in\FF_2^\ell\setminus\{0\}$, there exists at least one $v\in\FF_2^\ell$
with $\langle i,v\rangle=1$ and $\lvert v\rvert=1$, a sufficient condition for
$M_m\cdot\widehat{\Gamma}\geq\widehat{\Gamma}$ is that
\begin{align}\label{eq:overview_gamma_sufficient}
  A_v^m \widehat{\Gamma} & \geq (2^{2\ell-1}\epsilon^m n^m+1) \widehat{\Gamma}
\end{align}
for every $v\in\FF_2^\ell$ with $\lvert v\rvert=1$ .

Solving for each $A_v$ individually is analogous to the $\ell=1$ case. It is less clear, however,
how the above methods can be employed to solve jointly for all $A_v$. To this end we use the
combinatorial argument of~\cite{LL23a}, as follows. Let $F\subseteq\FF_2^{\ell\times n}$ and let
$\widehat{\Gamma}(X) = \One[X\in F]$. Consider the inequalities
in~\eqref{eq:overview_gamma_sufficient}: if $X\in\FF_2^{\ell\times n}\setminus F$, the right-hand
side is zero, while the left-hand side is non-negative.  Otherwise, $X\in F$ and the left-hand side
is the number of walks on the graph of $A_v$, of length $m$, that start at $X$ and end in $F$.

It remains to choose a set $F$ with minimal size and at least $2^{2\ell-1}\epsilon^m n^m+1$ many
returning walks. The symmetry of the problem suggests to seek $F$ among the \emph{configuration}
sets, i.e., the orbits of $\FF_2^{\ell\times n}$ with respect to the $S_n$-action.
In~\cref{lem:Ahpsi,lem:Apower}, we count
the returning walks for configurations. In \cref{sec:find_good_configs}, we choose configurations
that lead to the desired result.

\subsection{Completeness via Subspace Symmetric Dual LPs}

We now provide an overview of some ingredients and ideas in the completeness
proof from~\cref{sec:completeness}. As mentioned above, this new proof will take
place in the dual formulation of these hierarchies, as opposed to the proof of~\cite{CJJ23},
which takes place entirely in their primal formulation. Consequently, this new proof
will be useful in shedding new light on the structure of the dual.

Recall that our goal is to prove that the hierarchy~\eqref{eq:LPdual} is {\it exactly complete} at
level $n$: its optimum is the maximum $\lvert C\rvert^{\ell}$ for $C\in\Valid_n$, for every
$\ell\geq n$. Note that this will imply the same for the stronger partial Fourier hierarchy
of~\cite{LL23b} (see~\eqref{eq:pLPprimal}, \eqref{eq:symmpLPdual} and~\eqref{eq:pLPdual}).  Our
starting point will be the subspace symmetric formulation of these hierarchies from~\cite{CJJ23}. We
recall this formulation later, in~\eqref{eq:Mprimal}, and provide its dual in~\eqref{eq:Mdual}, but
we will not need their precise details in this high-level overview.

The proof proceeds as follows. If the hierarchy is indeed complete at level $\ell$,
then there exists a dual solution whose value is $q^{k\ell}$, where
$k\df\max\{\dim_{\FF_q}(C) \mid C\in \Valid_n\}$. 
However, constructing such a solution directly seems extremely hard.
Instead, we consider a weaker hierarchy, by replacing
the set $\Valid_{n}$, which depends on the code's distance,
with the set $\Valid_{n}^{\dim\leq k}$, which includes all linear
codes of dimension at most $k$. We observe that
\begin{align*}
    \Valid_{n} &\subseteq \Valid_{n}^{\dim\leq k},
    &
    \max \{\lvert C\rvert \mid C\in \Valid_{n}\} &= \max \{\lvert C\rvert \mid C \in \Valid_{n}^{\dim\leq k}\}
\end{align*}
We then proceed to analyze this weaker hierarchy since it suffices to prove
its completeness to deduce that the original hierarchy is also complete.
A key observation is that to obtain the desired tight objective value $q^{k\ell}$
several LP variables are forced to be zero. This will simplify the structure
of the dual, leading to a recurrence relating the value of the remaining variables.

\section{Lifting Dual Solutions}\label{sec:lifting}

In this section we show that dual solutions lift. That is, from a solution $h$ at a level $k$ of
value $V_h$, we can construct a natural solution at any level $\ell$ divisible by $k$ with value
$V_h^{\ell/k}$. Let us point out that in terms of values, it was already known
from~\cite[Corollary~6.6]{CJJ22} that the value of the hierarchy~\eqref{eq:LPprimal} at level $\ell$
was at most the $\ell/k$th power of its value at level $k$ (provided $k$ divides $\ell$); the main
contribution of this section is an explicit lift of dual solutions and the analogous result for the
partial Fourier hierarchy~\eqref{eq:pLPprimal}, which does not immediately follow from the results
of~\cite{CJJ22}.

\subsection{Further Symmetrization of the Dual}

Our first order of business is to use the $\GL_\ell(\FF_q)$-symmetry to simplify the dual
program. We start by recalling that the standard dual of the partial Fourier hierarchy
of~\eqref{eq:pLPprimal} is~\eqref{eq:pLPdual} below.

\begin{empheq}[box=\fbox]{equation}\label{eq:pLPdual}
  \begin{aligned}
    \text{Variables: }
    & \mathrlap{%
      h_{k,M}\colon\FF_q^{\ell\times n} \to \RR \;(k\in[\ell],M\in\GL_\ell(\FF_q),
      \beta\colon\FF_q^{\ell\times n}\to\RR%
    }
    \\
    \min \quad
    & 1 + \sum_{\mathclap{\substack{k\in[\ell]\\M\in\GL_\ell(\FF_q)}}} \cF_{k,M}(h_{k,M})(0)
    \\
    \text{s.t.} \quad
    & 1 + \sum_{\mathclap{\substack{k\in[\ell]\\M\in\GL_\ell(\FF_q)}}} \cF_{k,M}(h_{k,M})(X)
    + \beta(X) - \beta(-X)\leq 0
    & & \forall X\in\Valid_{n,\ell}\setminus\{0\}
    & (\text{Validity})
    \\
    & h_{k,M}(X) \geq 0
    & & \qquad\quad
    \mathllap{\forall X \in \FF_q^{\ell\times n}, \forall k\in[\ell], \forall M\in\GL_\ell(\FF_q)}
    & \mathllap{(\text{Non-negativity})}
  \end{aligned}
\end{empheq}

\begin{remark}\label{rmk:LPvspLP}
  It will also be useful to think of hierarchy \eqref{eq:LPdual} as a special case of~\eqref{eq:pLPdual} above. For
  this, note that every solution of~\eqref{eq:LPdual} yields a solution of~\eqref{eq:pLPdual} with
  the same value by setting $h_{\ell,I}\df 2^{n\ell}(\widehat{g} - \One_0)$ and setting all
  other $h_{k,M}$ to zero. Conversely, if $((h_{k,M})_{k,M},\beta)$ is a solution
  of~\eqref{eq:pLPdual} such that $h_{k,M}=0$ whenever $(k,M)\neq(\ell,I)$, then we can obtain a
  solution of~\eqref{eq:LPdual} of better or equal value by taking $g\df (1 +
  \widehat{h}_{\ell,I})/(1+2^{n\ell}h_{\ell,I}(0))$. Thus, hierarchy \eqref{eq:LPdual} is equivalent
  to~\eqref{eq:pLPdual} with the extra constraints that $h_{k,M}=0$ whenever $(k,M)\neq(\ell,I)$.
\end{remark}

We will now symmetrize~\eqref{eq:pLPdual} and pass to the Fourier basis, proving that it is
equivalent to~\eqref{eq:symmpLPdual}.

\begin{lemma}\label{lem:symmpLPdual}
  If $((h_{k,M})_{k,M},\beta)$ is a solution of~\eqref{eq:pLPdual}, then letting
  \begin{align*}
    g_k
    & \df
    \sum_{M\in\GL_\ell(\FF_q)}\cF_{k,M}(h_{k,M})\cdot M
    \qquad (k\in[\ell])
  \end{align*}
  yields a solution of~\eqref{eq:symmpLPdual} with the same value.

  Conversely, if $(g_k)_k$ is a solution of~\eqref{eq:symmpLPdual}, then letting
  \begin{align*}
    h_{k,M} & \df \frac{q^{kn}}{\lvert\GL_\ell(\FF_q)\rvert}\cdot\cF_{k,M}(g_k\cdot M)
    \qquad (k\in[\ell], M\in\GL_\ell(\FF_q)),
    \\
    \beta & \df 0,
  \end{align*}
  yields a solution of~\eqref{eq:pLPdual} with the same value.
\end{lemma}

\begin{proof}
  For the first direction, note that for $X\in\Valid_{n,\ell}$ (zero or not), we have
  \begin{align*}
    & \!\!\!\!\!\!
    1
    + \frac{1}{\lvert\GL_\ell(\FF_q)\rvert}
    \cdot\sum_{\substack{k\in[\ell]\\M\in\GL_\ell(\FF_q)}} (g_k\cdot M)(X)
    \\
    & =
    1
    + \frac{1}{\lvert\GL_\ell(\FF_q)\rvert}
    \cdot\sum_{k\in[\ell]}\sum_{M,N\in\GL_\ell(\FF_q)}(\cF_{k,N}(h_{k,N})\cdot (N\cdot M))(X)
    \\
    & =
    \frac{1}{\lvert\GL_\ell(\FF_q)\rvert}\cdot\sum_{M\in\GL_\ell(\FF_q)}
    \left(
    1
    + \sum_{k\in[\ell]}\sum_{N\in\GL_\ell(\FF_q)}
    \cF_{k,N}(h_{k,N})(M\cdot X)
    +
    \beta(M\cdot X)  - \beta(-M\cdot X)
    \right),
  \end{align*}
  where the last equality follows by a change of variables and since the $\beta$ contributions
  cancel out when we sum over $M$.

  Since $\Valid_{n,\ell}$ is $\GL_\ell(\FF_q)$-invariant, if $X\neq 0$, then the above is simply an
  average of the left-hand side of the validity constraints in~\eqref{eq:pLPdual}, so it must be
  non-positive.

  On the other hand, if $X=0$, then the first expression in the above is the objective value
  of~\eqref{eq:symmpLPdual} and the last expression is the objective of~\eqref{eq:pLPdual} (as both the
  average over $M$ goes away and the $\beta$ contributions cancel out since $M\cdot 0 = 0$).

  Finally, note that by~\eqref{eq:cFkMcFk}, we have
  \begin{align*}
    \cF_k(g_k)
    & =
    \sum_{M\in\GL_\ell(\FF_q)}\cF_k(\cF_{k,M}(h_{k,M})\cdot M)
    \\
    & =
    \sum_{M\in\GL_\ell(\FF_q)}\cF_{k,M}(\cF_{k,M}(h_{k,M}))\cdot M
    \\
    & =
    q^{-kn}\cdot\sum_{M\in\GL_\ell(\FF_q)} h_{k,M}\cdot R_k^{-1}\cdot M.
  \end{align*}
  Since $h_{k,M}\geq 0$ for every $M\in\GL_\ell(\FF_q)$, we conclude that $\cF_k(g_k)\geq 0$.

  \medskip

  We now prove the converse. Note that for $X\in\Valid_{n,\ell}$ (zero or not), we have
  \begin{align*}
    & \!\!\!\!\!\!
    1 + \sum_{\substack{k\in[\ell]\\M\in\GL_\ell(\FF_q)}} \cF_{k,M}(h_{k,M})(X) + \beta(X) - \beta(-X)
    \\
    & =
    1
    + \sum_{\substack{k\in[\ell]\\M\in\GL_\ell(\FF_q)}} \frac{q^{kn}}{\lvert\GL_\ell(\FF_q)\rvert}\cdot
    \cF_{k,M}(\cF_{k,M}(g_k\cdot M))(X)
    \\
    & =
    1
    + \frac{1}{\lvert\GL_\ell(\FF_q)\rvert}\cdot
    \sum_{\substack{k\in[\ell]\\M\in\GL_\ell(\FF_q)}} (g_k\cdot M\cdot R_k^{-1})(X),
  \end{align*}
  where the second equality follows from~\eqref{eq:cFkMcFk}.

  Since $\Valid_{n,\ell}$ is $\GL_\ell(\FF_q)$-invariant, if $X\neq 0$, then the above is non-positive as it
  is the left-hand side of a validity constraint in~\eqref{eq:symmpLPdual}.

  On the other hand, if $X=0$, then the first expression in the above is the objective value
  of~\eqref{eq:pLPdual} (as the $\beta$ contributions cancel out) and the last expression is the
  objective value of~\eqref{eq:symmpLPdual} (as the average over $M$ goes away in the latter since
  $M\cdot R_k^{-1}\cdot 0 = 0$).

  Finally, note that~\eqref{eq:cFkMcFk} implies
  \begin{align*}
    h_{k,M}
    & =
    \frac{q^{kn}}{\lvert\GL_\ell(\FF_q)\rvert}\cdot\cF_{k,M}(g_k\cdot M^{-1})
    =
    \frac{1}{\lvert\GL_\ell(\FF_q)\rvert}\cdot\cF_k(g_k)\cdot M^{-1}.
  \end{align*}
  Since $\cF_k(g_k)\geq 0$, we conclude that $h_{k,M}\geq 0$.
\end{proof}

\begin{remark}\label{rmk:fullFourier}
  Recalling from \cref{rmk:LPvspLP} that hierarchy \eqref{eq:LPdual} is equivalent to~\eqref{eq:pLPdual}
  with the extra constraints that $h_{k,M}=0$ whenever $(k,M)\neq(\ell,I)$, an analogue of
  \cref{lem:symmpLPdual} shows that the dual above is equivalent to~\eqref{eq:symmpLPdual} with the
  extra constraints that $g_k=0$ for every $k\in[\ell-1]$.
\end{remark}

\subsection{Basic Properties}

We now prove some basic combinatorial properties about matrices over $\FF_q$.

\begin{lemma}\label{lem:GLell}
  For a prime power $q$ and $\ell\in\NN$, the group
  \begin{align*}
    \GL_\ell(\FF_q)
    & \df
    \{M\in\FF_q^{\ell\times\ell} \mid \det(M)\neq 0\}
  \end{align*}
  has size exactly
  \begin{align*}
    (q-1)^\ell\cdot q^{\binom{\ell}{2}}\cdot\ell!_q
  \end{align*}
\end{lemma}

\begin{proof}
  By iteratively counting how many columns preserve linear independence, we get
  \begin{align*}
    \lvert\GL_\ell(\FF_q)\rvert
    & =
    \prod_{j=0}^{\ell-1} (q^\ell-q^j)
    =
    (q-1)^\ell\cdot q^{\binom{\ell}{2}}\cdot\prod_{j=0}^{\ell-1} [\ell-j]_q
    =
    (q-1)^\ell\cdot q^{\binom{\ell}{2}}\cdot\ell!_q.
    \qedhere
  \end{align*}
\end{proof}

\begin{definition}\label{def:Mqst}
  Let $q$ be a prime power, let $s,t,\ell,n\in\NN$ with $s\leq t\leq\ell\leq n$ and let
  $X\in\FF_q^{\ell\times n}$. We define
  \begin{align*}
    M_q^{s,t}(X)
    & \df
    \{M\in\GL_\ell(\FF_q) \mid
    (M\cdot X)_{1,\ldots,s}=0
    \land
    (M\cdot X)_{t+1,\ldots,\ell}=0
    \}.
  \end{align*}
  When $t=\ell$, we will use the shorthand notation $M_q^s(X)\df M_q^{s,\ell}(X)$.

  Furthermore, we define the \emph{marginal action} of $\GL_s(\FF_q)$ on $\GL_\ell(\FF_q)$ by
  \begin{align*}
    N\cdot M & \df
    \begin{pmatrix}
      N & 0\\
      0 & I
    \end{pmatrix}
    \cdot M
    \qquad (N\in\GL_s(\FF_q), M\in\GL_\ell(\FF_q))
  \end{align*}
  (on the right-hand side, the identity matrix is of order $\ell-s$ and the product is the usual
  matrix product).
\end{definition}

\begin{lemma}\label{lem:Mqst}
  Let $q$ be a prime power, let $s,t,\ell,n\in\NN$ with $s\leq t\leq\ell\leq n$ and let
  $X\in\FF_q^{\ell\times n}$. Then the following hold.

  \begin{enumerate}
  \item\label{lem:Mqst:inv} The sets $M_q^{0,t}(X)$ and $M_q^{s,t}(X)$ are $\GL_s(\FF_q)$-invariant.
  \item\label{lem:Mqst:dist} If $\rn{M}$ is picked uniformly at random in $M_q^{0,t}(X)$, then the
    distribution of $(\rn{M}\cdot X)_{1,\ldots,s}$ is $\GL_s(\FF_q)$-invariant.
  \item\label{lem:Mqst:size} For $z = s+\ell-t$ and $r\df\rk(X)$, we have
    \begin{align*}
      \lvert M_q^{s,t}(X)\rvert
      & =
      \lvert M_q^z(X)\rvert
      =
      (q-1)^\ell
      \cdot q^{\binom{\ell}{2}}
      \cdot (\ell-z)_{r,q}
      \cdot (\ell-r)!_q.
    \end{align*}
  \end{enumerate}
\end{lemma}

\begin{proof}
  \Cref{lem:Mqst:inv} follows since for every $N\in\GL_s(\FF_q)$ and every
  $M\in\GL_\ell(\FF_q)$, we have
  \begin{align*}
    \left(
    \begin{pmatrix}
      N & 0\\
      0 & I
    \end{pmatrix}
    \cdot
    M\cdot X
    \right)_{1,\ldots,s}
    =
    0
    & \iff
    (M\cdot X)_{1,\ldots,s} = 0,
    \\
    \left(
    \begin{pmatrix}
      N & 0\\
      0 & I
    \end{pmatrix}
    \cdot
    M\cdot X
    \right)_{t+1,\ldots,\ell}
    =
    0
    & \iff
    (M\cdot X)_{t+1,\ldots,\ell} = 0.
  \end{align*}

  \medskip

  For \cref{lem:Mqst:dist}, note that the distribution of $\rn{M}$ is
  $\GL_s(\FF_q)$-invariant. Thus, for every $N\in\GL_s(\FF_q)$, we have
  \begin{align*}
    N\cdot (\rn{M}\cdot X)_{1,\ldots,s}
    & =
    \left(
    \begin{pmatrix}
      N & 0\\
      0 & I
    \end{pmatrix}
    \cdot \rn{M}\cdot X
    \right)_{1,\ldots,s}
    \sim
    (\rn{M}\cdot X)_{1,\ldots,s}.
  \end{align*}

  \medskip

  It remains to prove \cref{lem:Mqst:size}.

  The fact that $\lvert M_q^{s,t}(X)\rvert = \lvert M_q^z(X)\rvert$ follows since there is a
  natural bijection between these sets obtained by permuting rows $s+1,\ldots,z$ with rows
  $t+1,\ldots,\ell$.

  Let us then compute the size of $M_q^z(X)$. First note that for every $N\in\GL_\ell(\FF_q)$,
  we have
  \begin{align*}
    M_q^z(N\cdot X)
    & =
    \{M\cdot N^{-1} \mid M\in M_q^z(X)\},
  \end{align*}
  so it suffices to show only the case when
  \begin{align*}
    X_1=e_1,X_2=e_2,\ldots,X_z=e_z,X_{z+1}=0,X_{z+2}=0,\ldots,X_\ell=0,
  \end{align*}
  where $e_i\in\FF_q^n$ is the $i$th canonical basis vector.

  By decomposing an element $M\in M_q^z(X)$ into blocks as
  \begin{align*}
    M & =
    \begin{pmatrix}
      A & B\\
      C & D
    \end{pmatrix}
  \end{align*}
  such that $A\in\FF_q^{z\times r}$, $B\in\FF_q^{z\times(n-r)}$, $C\in\FF_q^{(\ell-z)\times r}$ and
  $D\in\FF_q^{(\ell-z)\times(n-r)}$, we note that we must have $A = 0$, so we can count the elements
  of $M_q^z(X)$ by iteratively counting how many columns preserve linear independence to get
  \begin{align*}
    \lvert M_q^z(X)\rvert
    & =
    \left(\prod_{j=0}^{r-1}(q^{\ell-z}-q^j)\right)
    \cdot\prod_{j=r}^{\ell-1}(q^\ell-q^j)
    =
    (q-1)^\ell
    \cdot q^{\binom{\ell}{2}}
    \cdot (\ell-z)_{r,q}
    \cdot (\ell-r)!_q,
  \end{align*}
  as desired.
\end{proof}

\subsection{The Lifts}

We now have all the ingredients to lift dual solutions. We start with a warm-up by lifting solutions
from level $1$ to level $\ell$. The bold reader should feel free to skip directly to
\cref{thm:lift}.

\begin{proposition}\label{prop:lift}
  Let $q$ be a prime power. If $h$ is a solution of~\eqref{eq:symmpLPdual} with $\ell=1$, then for
  every $\ell\in[n]$, letting
  \begin{align*}
    g_1 & \df g_2 \df \cdots \df g_{\ell-1} \df 0,
    &
    g_\ell(X)
    & \df
    \sum_{t\in[\ell]} (1+h(0))^{\ell-t}\cdot h(X_1)\cdot\One[X_{t+1,\ldots,\ell} = 0]
  \end{align*}
  gives a solution of~\eqref{eq:symmpLPdual} whose objective value is the $\ell$th power of the
  objective value of $h$, i.e., we have
  \begin{align*}
    1 + \sum_{u\in[\ell]} g_u(0) & = (1+h(0))^\ell.
  \end{align*}
\end{proposition}

\begin{proof}
  Clearly $\cF_u(g_u)=0$ for every $u\in[\ell-1]$. Also, for every $X\in\FF_q^{\ell\times n}$, we
  have
  \begin{align*}
    \cF_\ell(g_\ell)(X)
    & =
    \sum_{t\in[\ell]}
    (1+h(0))^{\ell-t}
    \cdot\widehat{h}(X_1)
    \cdot\One[X_{2,\ldots,t}=0]
    \cdot 2^{-n(\ell-t)},
  \end{align*}
  which is non-negative as $\widehat{h}=\cF_1(h)\geq 0$.

  Let us now show the validity constraints. Let $X\in\Valid\setminus\{0\}$ and let $r\df\rk(X)$. We
  need to show that
  \begin{align*}
    1
    + \frac{1}{\lvert\GL_\ell(\FF_q)\rvert}\cdot
    \sum_{u\in[\ell]}\sum_{M\in\GL_\ell(\FF_q)} (g_u\cdot M)(X)
    \leq 0,
  \end{align*}
  which is equivalent to
  \begin{align*}
    \sum_{t\in[\ell]}
    (1+h(0))^{\ell-t}
    \cdot
    \frac{1}{\lvert\GL_\ell(\FF_q)\rvert}\cdot
    \sum_{M\in\GL_\ell(\FF_q)}
    h((M\cdot X)_1)\cdot\One[(M\cdot X)_{t+1,\ldots,\ell}=0]
    \leq -1,
  \end{align*}
  which in turn is equivalent to
  \begin{align}\label{eq:objective}
    \sum_{t\in[\ell]}
    (1+h(0))^{\ell-t}
    \cdot
    \EE_{\rn{M}\sim U(\GL_\ell(\FF_q))}[
      h((\rn{M}\cdot X)_1)\cdot\One[(\rn{M}\cdot X)_{t+1,\ldots,\ell}=0]
    ] & \leq -1,
  \end{align}
  where $U(\GL_\ell(\FF_q))$ is the uniform distribution on $\GL_\ell(\FF_q)$.

  To prove the above, fix $t\in[\ell]$ and let us study the expression $h((\rn{M}\cdot
  X)_1)\cdot\One[(\rn{M}\cdot X)_{t+1,\ldots,\ell}=0]$.

  First, recalling \cref{def:Mqst}, we partition $\GL_\ell(\FF_q)$ into the sets
  \begin{align*}
    M_q^{1,t}(X),
    & &
    M_q^{0,t}(X)\setminus M_q^{1,t}(X),
    & &
    \GL_\ell(\FF_q)\setminus M_q^{0,t}(X),
  \end{align*}
  which by \cref{lem:GLell,lem:Mqst} have sizes
  \begin{equation}\label{eq:sizes}
    \begin{gathered}
      (q-1)^\ell
      \cdot q^{\binom{\ell}{2}}
      \cdot (t-1)_{r,q}
      \cdot (\ell-r)!_q,
      \\
      (q-1)^\ell
      \cdot q^{\binom{\ell}{2}}
      \cdot ((t)_{r,q} - (t-1)_{r,q})
      \cdot (\ell-r)!_q,
      \\
      (q-1)^\ell
      \cdot q^{\binom{\ell}{2}}
      \cdot ((\ell)_{r,q} - (t)_{r,q})
      \cdot (\ell-r)!_q,
    \end{gathered}
  \end{equation}
  respectively.

  Note that $\One[(\rn{M}\cdot X)_{t+1,\ldots,\ell}=0]$ only takes non-zero values when $\rn{M}$ is
  in one of the first two sets.

  Clearly, if we condition on $\rn{M}\in M_q^{1,t}(X)$, then $(\rn{M}\cdot X)_1 = 0$.

  On the other hand, by \cref{lem:Mqst:inv,lem:Mqst:dist} of \cref{lem:Mqst}, we know that the
  conditional distribution of $(\rn{M}\cdot X)_1$ given $\rn{M}\in M_q^{0,t}(X)\setminus
  M_q^{1,t}(X)$ is $\GL_1(\FF_q)$-invariant. In particular, this implies that if $\rn{N}\sim
  U(\GL_1(\FF_q))$ is independent from $\rn{M}$, then the conditional distributions of $(\rn{M}\cdot
  X)_1$ and $\rn{N}\cdot(\rn{M}\cdot X)_1$ given $\rn{M}\in M_q^{0,t}(X)\setminus M_q^{1,t}(X)$ are
  the same. Finally, since $X\in\Valid\setminus\{0\}$, we know that when we condition on $\rn{M}\in
  M_q^{0,t}(X)\setminus M_q^{1,t}(X)$, then $(\rn{M}\cdot X)_1$ is always an element of
  $\Valid\setminus\{0\}$. Thus we conclude that
  \begin{gather*}
    \EE_{\rn{M}}[
      h((\rn{M}\cdot X)_1)\cdot\One[(\rn{M}\cdot X)_{t+1,\ldots,\ell}=0]
      \given
      \rn{M}\in M_q^{1,t}(X)
    ]
    =
    h(0),
    \\
    \begin{multlined}
    \EE_{\rn{M}}[
      h((\rn{M}\cdot X)_1)\cdot\One[(\rn{M}\cdot X)_{t+1,\ldots,\ell}=0]
      \given
      \rn{M}\in M_q^{0,t}(X)\setminus M_q^{1,t}(X)
    ]
    \\
    =
    \EE_{\rn{M}}[
      \EE_{\rn{N}}[
        h(\rn{N}\cdot(\rn{M}\cdot X)_1)
      ]
      \given
      \rn{M}\in M_q^{0,t}(X)\setminus M_q^{1,t}(X)
    ]
    \leq
    -1,
    \end{multlined}
    \\
    \EE_{\rn{M}}[
      h((\rn{M}\cdot X)_1)\cdot\One[(\rn{M}\cdot X)_{t+1,\ldots,\ell}=0]
      \given
      \rn{M}\in \GL_\ell(\FF_q)\setminus M_q^{0,t}(X)
    ]
    =
    0,
  \end{gather*}
  where the inequality follows from the validity constraints for $h$.

  Thus, we get
  \begin{align*}
    \EE_{\rn{M}\sim U(\GL_\ell(\FF_q))}[
      (\rn{M}\cdot X)_1\cdot\One[(\rn{M}\cdot X)_{t+1,\ldots,\ell}=0]
    ]
    & \leq
    \frac{
      \lvert M_q^{1,t}(X)\rvert\cdot h(0)
      - \lvert M_q^{0,t}(X)\setminus M_q^{1,t}(X)\rvert
    }{
      \lvert\GL_\ell(\FF_q)\rvert\cdot
    }
    \\
    & =
    \frac{(t-1)_{r,q}\cdot (1+h(0)) - (t)_{r,q}}{(\ell)_{r,q}},
  \end{align*}
  where the equality follows from \cref{lem:GLell} and~\eqref{eq:sizes}.

  Recalling that our goal is to show~\eqref{eq:objective}, we note that
  \begin{align*}
    & \!\!\!\!\!\!
    \sum_{t\in[\ell]}
    (1+h(0))^{\ell-t}
    \cdot
    \EE_{\rn{M}\sim U(\GL_\ell(\FF_q))}[
      (\rn{M}\cdot X)_1\cdot\One[(\rn{M}\cdot X)_{t+1,\ldots,\ell}=0]
    ]
    \\
    & \leq
    \frac{1}{(\ell)_{r,q}}
    \sum_{t\in[\ell]}
    (1+h(0))^{\ell-t}
    \cdot((t-1)_{r,q}\cdot (1+h(0)) - (t)_{r,q})
    \\
    & =
    \frac{(1+h(0))^\ell\cdot (0)_{r,q} - (\ell)_{r,q}}{(\ell)_{r,q}}
    =
    -1,
  \end{align*}
  where the first equality follows since the sum telescopes. Thus, $g$ is a feasible solution.

  It remains to compute the value of $g$. Note that
  \begin{align*}
    1 + \sum_{u\in[\ell]} g_u(0)
    & =
    1 +
    \sum_{t\in[\ell]}
    (1+h(0))^{\ell-t}
    \cdot h(0)
    \\
    & =
    1 +
    \sum_{t\in[\ell]}
    (1+h(0))^{\ell-t}
    \cdot ((1+h(0)) - 1)
    =
    (1+h(0))^\ell,
  \end{align*}
  where the last equality follows since the sum telescopes.
\end{proof}

\begin{remark}
  Note that since the lift in \cref{prop:lift} sets all $g_u$ with $u < \ell$ to $0$, it follows
  that this is also a lift of the dual of the full Fourier hierarchy (see~\cref{rmk:fullFourier}).
\end{remark}

We now prove the more general lift from level $k$ to level $\ell$ under the assumption that $k$
divides $\ell$. We point out that when we take $k=1$ in \cref{thm:lift} below, we recover
\cref{prop:lift}, except for the fact that the constructed solution has coordinates
slightly permuted so that it is appropriately compatible with the partial Fourier.

\begin{theorem}\label{thm:lift}
  Let $q$ be a prime power and $k\in\NN_+$. If $h$ is a solution of~\eqref{eq:symmpLPdual} with
  $\ell=k$ and objective value $V_h\df 1+\sum_{u\in[k]} h_u(0)$, then for every $\ell\in[n]$ divisible
  by $k$, letting
  \begin{align*}
    g_u(X) & \df
    \begin{dcases*}
      0,
      & if $u\leq\ell-k$,
      \\
      \sum_{t=0}^{\ell/k-1}
      V_h^t
      \cdot h_{u-\ell+k}(X_{\ell-k+1,\ldots,\ell})
      \cdot\One[X_{1,\ldots,kt}=0],
      & otherwise,
    \end{dcases*}
  \end{align*}
  gives a solution of~\eqref{eq:symmpLPdual} whose objective value is the $(\ell/k)$th power of the
  objective value of $h$, i.e., we have
  \begin{align*}
    1 + \sum_{u\in[\ell]} g_u(0) & = V_h^{\ell/k} = \left(1 + \sum_{u\in[k]} h_u(0)\right)^{\ell/k}.
  \end{align*}
\end{theorem}

\begin{proof}
  Clearly $\cF_u(g_u)=0$ for every $u\in[\ell-k]$. Also, for every $X\in\FF_q^{\ell\times n}$, we
  have
  \begin{align*}
    \cF_\ell(g_\ell)(X)
    & =
    \sum_{t=0}^{\ell/k-1}
    V_h^t
    \cdot\cF_{u-\ell+k}(h_{u-\ell+k}(X_{\ell-k+1,\ldots,\ell}))
    \cdot\One[X_{kt+1,\ldots,\ell-k}=0]
    \cdot 2^{-nkt},
  \end{align*}
  which is non-negative as $\cF_u(h_u)\geq 0$ for every $u\in[k]$.

  Let us now show the validity constraints. Let $X\in\Valid\setminus\{0\}$ and let $r\df\rk(X)$. We
  need to show that
  \begin{align*}
    1
    +
    \frac{1}{\lvert\GL_\ell(\FF_q)\rvert}\cdot
    \sum_{u\in[\ell]} (g_u\cdot M)(X)
    \leq 0,
  \end{align*}
  which is equivalent to
  \begin{align*}
    \sum_{t=0}^{\ell/k-1}
    V_h^t
    \sum_{u=\ell-k+1}^\ell
    \frac{1}{\lvert\GL_\ell(\FF_q)\rvert}
    \sum_{M\in\GL_\ell(\FF_q)}
    h_{u-\ell+k}((M\cdot X)_{\ell-k+1,\ldots,\ell})
    \cdot\One[(M\cdot X)_{1,\ldots,kt}=0]
    \leq -1.
  \end{align*}
  By permuting the rows of the resulting matrix in the expression above, we see that the above is
  equivalent to
  \begin{align}\label{eq:objectiveho}
    \sum_{t=0}^{\ell/k-1}
    V_h^t
    \sum_{u\in[k]}
    \EE_{\rn{M}\sim U(\GL_\ell(\FF_q))}[
      h_u((M\cdot X)_{1,\ldots,k})
      \cdot\One[(M\cdot X)_{\ell-kt+1,\ldots,\ell}=0]
    ]
    \leq -1,
  \end{align}
  where $U(\GL_\ell(\FF_q))$ is the uniform distribution on $\GL_\ell(\FF_q)$.

  To prove the above, fix $t\in\{0,\ldots,\ell/k-1\}$ and let us study the expression
  \begin{align*}
    \sum_{u\in[k]}
    h_u((M\cdot X)_{1,\ldots,k})
    \cdot\One[(M\cdot X)_{\ell-kt+1,\ldots,\ell}=0],
  \end{align*}
  where $\rn{M}\sim U(\GL_\ell(\FF_q))$.

  First, recalling \cref{def:Mqst}, we partition $\GL_\ell(\FF_q)$ into the sets
  \begin{align*}
    M_q^{k,\ell-kt}(X),
    & &
    M_q^{0,\ell-kt}(X)\setminus M_q^{k,\ell-kt}(X),
    & &
    \GL_\ell(\FF_q)\setminus M_q^{0,\ell-kt}(X),
  \end{align*}
  which by \cref{lem:GLell,lem:Mqst} have sizes
  \begin{equation}\label{eq:sizesho}
    \begin{gathered}
      (q-1)^\ell
      \cdot q^{\binom{\ell}{2}}
      \cdot (\ell-k(t+1))_{r,q}
      \cdot (\ell-r)!_q,
      \\
      (q-1)^\ell
      \cdot q^{\binom{\ell}{2}}
      \cdot ((\ell-kt)_{r,q} - (\ell-k(t+1))_{r,q})
      \cdot (\ell-r)!_q,
      \\
      (q-1)^\ell
      \cdot q^{\binom{\ell}{2}}
      \cdot ((\ell)_{r,q} - (\ell-kt))_{r,q})
      \cdot (\ell-r)!_q,
    \end{gathered}
  \end{equation}
  respectively.

  Note that $\One[(\rn{M}\cdot X)_{\ell-kt+1,\ldots,\ell}=0]$ only takes non-zero values when
  $\rn{M}$ is in one of the first two sets.

  Clearly, if we condition on $\rn{M}\in M_q^{k,\ell-kt}(X)$, then $(\rn{M}\cdot X)_{1,\ldots,k}=0$.

  On the other hand, by \cref{lem:Mqst:inv,lem:Mqst:dist} of \cref{lem:Mqst}, we know that the
  conditional distribution of $(\rn{M}\cdot X)_{1,\ldots,k}$ given $\rn{M}\in
  M_q^{0,\ell-kt}(X)\setminus M_q^{k,\ell-kt}(X)$ is $\GL_k(\FF_q)$-invariant. In particular, this
  implies that if $\rn{N}\sim U(\GL_k(\FF_q))$ is independent from $\rn{M}$, then the conditional
  distributions of $(\rn{M}\cdot X)_{1,\ldots,k}$ and $\rn{N}\cdot(\rn{M}\cdot X)_{1,\ldots,k}$
  given $\rn{M}\in M_q^{0,\ell-kt}(X)\setminus M_q^{k,\ell-kt}(X)$ are the same. Finally, since
  $X\in\Valid\setminus\{0\}$, we know that when we condition on $\rn{M}\in
  M_q^{0,\ell-kt}(X)\setminus M_q^{k,\ell-kt}(X)$, then $(\rn{M}\cdot X)_{1,\ldots,k}$ is always an
  element of $\Valid\setminus\{0\}$. Thus we conclude that
  \begin{gather*}
    \sum_{u\in[k]}
    \EE_{\rn{M}}[
      h_u((M\cdot X)_{1,\ldots,u})
      \cdot\One[(M\cdot X)_{\ell-kt+1,\ldots,\ell}=0]
      \given
      \rn{M}\in M_q^{k,\ell-kt}(X)
    ]
    =
    \sum_{u\in[k]} h_u(0)
    =
    V_h-1,
    \\
    \begin{multlined}
      \sum_{u\in[k]}
      \EE_{\rn{M}}[
        h_u((M\cdot X)_{1,\ldots,u})
        \cdot\One[(M\cdot X)_{\ell-kt+1,\ldots,\ell}=0]
        \given
        \rn{M}\in M_q^{0,\ell-kt}(X)\setminus M_q^{k,\ell-kt}(X)
      ]
      \\
      =
      \sum_{u\in[k]}
      \EE_{\rn{M}}[
        \EE_{\rn{N}}[
          h_u(\rn{N}\cdot(M\cdot X)_{1,\ldots,u})
          \cdot\One[(M\cdot X)_{\ell-kt+1,\ldots,\ell}=0]
        ]
        \given
        \rn{M}\in M_q^{0,\ell-kt}(X)\setminus M_q^{k,\ell-kt}(X)
      ]
      \\
      \leq
      -1,
    \end{multlined}
    \\
    \sum_{u\in[k]}
    \EE_{\rn{M}}[
      h_u((M\cdot X)_{1,\ldots,u})
      \cdot\One[(M\cdot X)_{\ell-kt+1,\ldots,\ell}=0]
      \given
      \rn{M}\in \GL_\ell(\FF_q)\setminus M_q^{0,\ell-kt}(X)
    ]
    =
    0,
  \end{gather*}
  where the inequality follows from the validity constraints for $h$.

  Thus, we get
  \begin{align*}
    & \!\!\!\!\!\!
    \sum_{u\in[k]}
    \EE_{\rn{M}\sim U(\GL_\ell(\FF_q))}[
      h_u((M\cdot X)_{1,\ldots,u})
      \cdot\One[(M\cdot X)_{\ell-kt+1,\ldots,\ell}=0]
    ]
    \\
    & \leq
    \frac{
      \lvert M_q^{k,\ell-kt}(X)\rvert\cdot (V_h-1)
      - \lvert M_q^{0,\ell-kt}(X)\setminus M_q^{k,\ell-kt}(X)\rvert
    }{
      \lvert\GL_\ell(\FF_q)\rvert
    }
    \\
    & =
    \frac{(\ell-k(t+1))_{r,q}\cdot V_h - (\ell-kt)_{r,q}}{(\ell)_{r,q}},
  \end{align*}
  where the equality follows from \cref{lem:GLell} and~\eqref{eq:sizesho}.

  Recalling that our goal is to show~\eqref{eq:objective}, we note that
  \begin{align*}
    & \!\!\!\!\!\!
    \sum_{t=0}^{\ell/k-1}
    V_h^t
    \sum_{u\in[k]}
    \EE_{\rn{M}\sim U(\GL_\ell(\FF_q))}[
      h_u((M\cdot X)_{1,\ldots,u})
      \cdot\One[(M\cdot X)_{\ell-kt+1,\ldots,\ell}=0]
    ]
    \\
    & \leq
    \frac{1}{(\ell)_{r,q}}
    \sum_{t=0}^{\ell/k-1}
    V_h^t
    \cdot((\ell-k(t+1))_{r,q} V_h - (\ell-kt)_{r,q})
    \\
    & =
    \frac{V_h^{\ell/k}\cdot (0)_{r,q} - (\ell)_{r,q}}{(\ell)_{r,q}}
    =
    -1,
  \end{align*}
  where the first equality follows since the sum telescopes. Thus, $g$ is a feasible solution.

  It remains to compute the value of $g$. Note that
  \begin{align*}
    1 + \sum_{u\in[\ell]} g_u(0)
    & =
    1 +
    \sum_{t=0}^{\ell/k-1}
    V_h^t
    \sum_{u=\ell-k+1}^\ell
    h_{u-\ell+k}(0)
    =
    1 +
    \sum_{t=0}^{\ell/k-1} V_h^t(V_h-1)
    =
    V_h^{\ell/k},
  \end{align*}
  where the last equality follows since the sum telescopes.
\end{proof}

\section{Completeness via Subspace Symmetric Dual LPs}\label{sec:completeness}

We will now give a new proof that the hierarchy is complete, i.e., it recovers the true
size of a code at level $\ell\geq n$. For this proof, we recall yet another formulation of the
hierarchy from~\cite{CJJ23}.

Instead of symmetrizing~\eqref{eq:LPprimal} under the action of $S_n$, we recall that
$\GL_\ell(\FF_q)$ also acts on $\FF_q^{\ell\times n}$ by left-multiplication and observe
that~\eqref{eq:LPprimal} is also $\GL_\ell(\FF_q)$-symmetric. Inspired by terminology from
Sum-of-Squares algorithms, given a $\GL_\ell(\FF_q)$-symmetric solution $f$, for each $S\in
L_{\FF_q}(\FF_q^n)$, we define the notation
\begin{align*}
  \pPr[S\subseteq\rn{\widetilde{C}}] & \df f(X)
\end{align*}
for any $X\in\FF_q^{\ell\times n}$ with $\linspan(\{X_1,\ldots,X_\ell\})=S$ and interpret this as a
pseudo-probability that a pseudo-random variable $\rn{\widetilde{C}}$ over $L_{\FF_q}(\FF_q^n)$
contains $S$. Computing the pseudo-probabilities $\pPr[S]\df\pPr[S=\rn{\widetilde{C}}]$
amounts to a \Mobius\ inversion on the poset $L_{\FF_q}(\FF_q^n)$ under the inclusion partial
order. At levels $\ell\geq n$ and when $\Valid_n$ is closed under taking subspaces\footnote{It is
  possible to make this \Mobius\ inversion at lower levels and without the closure under subspaces
  assumption, but it yields more complicated constraints. Since our completeness result will only
  hold for levels $\ell\geq n$ anyway, we opt for the simpler formulation instead.}, this yields the
formulation in~\eqref{eq:Mprimal}, whose dual is~\eqref{eq:Mdual}; a code $C\in\Valid_n$ yields a
solution $\pPr_C[S]\df\One[S=C]$ of~\eqref{eq:Mprimal}, whose value is $\lvert C\rvert^\ell$. The
first completeness at levels $\ell\geq n$ of~\cite{CJJ23} was based on the primal
formulation~\eqref{eq:Mprimal} and crucially relied on the fact that non-negative solutions
to~\eqref{eq:Mprimal} are convex combinations of true solutions.

\begin{empheq}[box=\fbox]{gather}\label{eq:Mprimal}
  \begin{aligned}
    \text{Variables: } 
    & \mathrlap{(\pPr[S] \mid S\in L_{\FF_q}(\FF_q^n))}
    \\
    \max \qquad
    & \sum_{S\in L_{\FF_q}(\FF_q^n)} \lvert S\rvert^\ell \pPr[S]
    \\
    \text{s.t.} \qquad
    & \sum_{S\in L_{\FF_q}(\FF_q^n)} \pPr[S] = 1
    & &
    & (\text{Normalization})
    \\
    & \pPr[S] = 0
    & & \forall S\in L_{\FF_q}(\FF_q^n)\setminus\Valid_n
    & (\text{Validity})
    \\
    & \sum_{\substack{S\in L_{\FF_q}(\FF_q^n)\\ S\subseteq U}} \lvert S\rvert^\ell \pPr[S] \geq 0
    & & \forall U\in L_{\FF_q}(\FF_q^n)
    & (\text{Downward sums})
    \\
    & \sum_{\substack{S\in L_{\FF_q}(\FF_q^n)\\ U\subseteq S}} \pPr[S] \geq 0
    & & \forall U\in L_{\FF_q}(\FF_q^n)
    & (\text{Upward sums})
  \end{aligned}
\end{empheq}

\begin{empheq}[box=\fbox]{gather}\label{eq:Mdual}
  \begin{aligned}
    \text{Variables: }
    & \mathrlap{\alpha\in\RR, \beta,\gamma\colon L_{\FF_q}(\FF_q^n)\to\RR}
    \\
    \min \quad
    & \alpha
    \\
    \text{s.t.} \quad
    &
    \alpha
    =
    \lvert S\rvert^\ell
    + \lvert S\rvert^\ell\sum_{\mathclap{\substack{T\in L_{\FF_q}(\FF_q^n)\\S\leq T}}} \beta(T)
    + \sum_{\mathclap{\substack{T\in L_{\FF_q}(\FF_q^n)\\T\leq S}}} \gamma(T)
    & & \forall S\in\Valid_n
    & (\text{Equality to objective})
    \\
    & \beta(S) \geq 0
    & & \forall S\in L_{\FF_q}(\FF_q^n)
    & (\text{$\beta$ non-negativity})
    \\
    & \gamma(S) \geq 0
    & & \forall S\in L_{\FF_q}(\FF_q^n)
    & (\text{$\gamma$ non-negativity})
  \end{aligned}
\end{empheq}

It will also be convienient to define for every $k\in\NN$ the set
\begin{align*}
  \Valid_n^{\dim\leq k}
  & \df
  \{S\in L_{\FF_q}(\FF_q^n) \mid \dim_{\FF_q}(S)\leq k\}.
\end{align*}

It is clear that for any $\Valid_n\subseteq L_{\FF_q}(\FF_q^n)$ non-empty, if
$k\df\max\{\dim_{\FF_q}(S) \mid S\in\Valid_n\}$, then $\Valid_n\subseteq\Valid_n^{\dim\leq k}$. We
will show completeness of~\eqref{eq:Mdual} for valid sets of the form $\Valid_n^{\dim\leq k}$
($k\in\NN$) and leverage this to show completeness for arbitrary non-empty valid sets
$\Valid_n\subseteq L_{\FF_q}(\FF_q^n)$ that are closed under taking subspaces. We start with the
following key observation.

\noindent\textbf{Key observation:} With valid set $\Valid_n^{\dim\leq k}$, at completeness levels
(i.e., $\ell \ge n$), we must have $\alpha = q^{k\ell}$, and, for the dual to
achieve this optimum value, many variables $\beta(S)$ and $\gamma(S)$ will need to be zero. This
will greatly simplify the dual LP allowing us to establish a recurrence to determine bounds on the
remaining variables proving that they can be taken to be nonnegative thereby implying the
feasibility of the solution.

\begin{theorem}[Exact Completeness from the Dual]\label{thm:completeness_from_dual}
  For every $\ell\geq n$ and every $\Valid_n\subseteq L_{\FF_q}(\FF_q^n)$ non-empty and closed under
  taking subspaces, the optimum value of~\eqref{eq:Mdual} is $q^{\ell k}$, where
  \begin{align*}
    k & \df \max\{\dim_{\FF_q}(S) \mid S\in\Valid_n\}.
  \end{align*}
\end{theorem}

\begin{proof}
  Let us make the key observation above formal. First note that since
  $\Valid_n\subseteq\Valid_n^{\dim\leq k}$, it follows that~\eqref{eq:Mdual} with $\Valid_n$ has
  less constraints than the same program with $\Valid_n^{\dim\leq k}$, so it suffices to produce a
  feasible solution for~\eqref{eq:Mdual} with $\Valid_n^{\dim\leq k}$ whose value is $\alpha\df
  q^{\ell k}$. Since for every $S\in L_{\FF_q}(\FF_q^n)$ with $\dim_{\FF_q}(S)=k$ we have
  \begin{align*}
    \alpha
    =
    q^{\ell k}
    & =
    \lvert S\rvert^\ell
    + \lvert S\rvert^\ell \sum_{\substack{T\in L_{\FF_q}(\FF_q^n)\\S\subseteq T}} \beta(T)
    + \sum_{\substack{T\in L_{\FF_q}(\FF_q^n)\\ T\subseteq S}} \gamma(T)
    \\
    & =
    \lvert \FF_q\rvert^{\ell k}
    + \lvert\FF_q\rvert^{\ell k} \sum_{\substack{T\in L_{\FF_q}(\FF_q^n)\\S\subseteq T}} \beta(T)
    + \sum_{\substack{T\in L_{\FF_q}(\FF_q^n)\\T\subseteq S}} \gamma(T)
  \end{align*}
  and both $\beta$ and $\gamma$ must be non-negative, we must have $\beta(T) = 0$ whenever
  $\dim_{\FF_q}(T) \geq k$ and  $\gamma(T)=0$ whenever
  $\dim_{\FF_q}(T) \leq k$.

  Let us in fact set $\gamma(T)=0$ for every $T\in L_{\FF_q}(\FF_q^n)$. For $\beta$, it will be
  convenient (and sufficient) to consider $\beta(T) = \widetilde{\beta}_{\dim_{\FF_q}(T)}$, namely, these
  variables will only depend on the dimension. Then for a space $S\in L_{\FF_q}(\FF_q^n)$ of
  dimension $s$, the equality to objective constraint reads
  \begin{align*}
    \alpha = q^{\ell k}
    & =
    \lvert S\rvert^\ell
    + \lvert S\rvert^\ell \sum_{i=\dim_{\FF_q}(S)}^n
    \sum_{\substack{T\in L_{\FF_q}(\FF_q^n)\\ S\subseteq T\\\dim_{\FF_q}(T)=i}} \widetilde{\beta}_i
    \\
    & =
    q^{\ell s}
    + q^{\ell s} \sum_{i=s}^{k-1}
    \sum_{\substack{T\in L_{\FF_q}(\FF_q^n)\\ S\subseteq T\\\dim_{\FF_q}(T)=i}} \widetilde{\beta}_i
    & & (\text{Since $\widetilde{\beta}_i =0$ whenever $i \geq k$.})
    \\
    & =
    q^{\ell s}
    + q^{\ell s} \sum_{i=s}^{k-1} \binom{n-s}{i-s}_q \widetilde{\beta}_i.
  \end{align*}

  Thus, to satisfy all equality to objective constraints, the following recurrence must hold for
  every $s\in\{0,\ldots,k-1\}$:
  \begin{align}\label{eq:widetildebeta}
    \widetilde{\beta}_s
    & =
    q^{\ell(k-s)} - 1 - \sum_{i=s+1}^{k-1} \binom{n-s}{i-s}_q \widetilde{\beta}_i.
  \end{align}

  Our objective is then to prove by reverse induction in $s\in\{0,\ldots,k-1\}$ that defining
  $\widetilde{\beta}$ by~\eqref{eq:widetildebeta} above yields $\widetilde{\beta}_s\geq 0$ for every
  $s\in\{0,\ldots,k-1\}$.

  First note that~\eqref{eq:widetildebeta} for $s=k-1$ yields $\widetilde{\beta}_{k-1}=q^\ell-1\geq
  0$. Suppose now that $s\in\{0,\ldots,k-2\}$ and note that using~\eqref{eq:widetildebeta} for
  $\widetilde{\beta}_{s+1}$ in its version for $\widetilde{\beta}_s$, we get
  \begin{align*}
    \widetilde{\beta}_s
    & =
    q^{\ell(k-s)} - 1 - \sum_{i=s+2}^{k-1} \binom{n-s}{i-s}_q \widetilde{\beta}_i
    - \binom{n-s}{1}_q\left(
    q^{\ell(k-s-1)} - 1 - \sum_{i=s+2}^{k-1} \binom{n-s-1}{i-s-1}_q \widetilde{\beta}_i
    \right)
    \\
    & =
    q^{\ell(k-s)}\left(1 - \frac{[n-s]_q}{q^\ell}\right)
    + [n-s]_q - 1
    + \sum_{i=s+2}^{k-1}\left([n-s]_q\binom{n-s-1}{i-s-1}_q - \binom{n-s}{i-s}_q\right)\widetilde{\beta}_i
    & \geq
    0,
  \end{align*}
  where the inequality follows since
  \begin{align*}
    1 - \frac{[n-s]_q}{q^\ell}
    & \geq
    1 - q^{n-s-\ell}
    \geq
    0
    & & (\text{since $\ell\geq n$}),
    \\
    [n-s]_q - 1
    & \geq 0
    & & (\text{since $s\leq k-2 < n$}),
    \\
    [n-s]_q\binom{n-s-1}{i-s-1}_q - \binom{n-s}{i-s}_q
    & =
    \binom{n-s}{i-s}_q([i-s]_q - 1)
    \geq
    0
    & & (\text{for every $i\geq s+2$}),
  \end{align*}
  and since inductively we have $\widetilde{\beta}_i\geq 0$ for every $i\geq s+2$.

  Thus, we conclude that setting
  \begin{align*}
    \alpha & \df q^{\ell k}, &
    \beta(T) & \df \widetilde{\beta}_{\dim_{\FF_q}(T)}, &
    \gamma(T) & \df 0,
  \end{align*}
  (where $\widetilde{\beta}_s$ is given recursively by~\eqref{eq:widetildebeta} for
  $s\in\{0,\ldots,k-1\}$ and is zero when $s\geq k$) yields a feasible solution of~\eqref{eq:Mdual}
  (for both $\Valid_n$ and $\Valid_n^{\dim\leq k}$) whose value is $q^{\ell k}$.
\end{proof}

\section{Spectral-based Dual solutions for Balanced codes}\label{sec:spectral_construction}

In this section, we construct a spectral-based solution at level $\ell$ for $\epsilon$-balanced codes over
$\FF_2$ whose values are comparable with the MRRW solution. The set of \emph{(linear)
  $\epsilon$-balanced codes (over $\FF_2^n$)} is defined as
\begin{align*}
  \Valid_n^\epsilon
  & \df
  \left\{C\in L_{\FF_2}(\FF_2^n) \;\middle\vert\;
  \forall x\in C\setminus\{0\},
  \left((1-\epsilon)\frac{n}{2}
  \leq
  \lvert x\rvert
  \leq
  (1+\epsilon)\frac{n}{2}\right)
  \right\},
\end{align*}
so we have
\begin{align*}
  \Valid_{n,\ell}^\epsilon
  & =
  \left\{X\in\FF_2^{\ell\times n} \;\middle\vert\;
  \forall u\in\FF_2^\ell,
  \left(
  uX\neq 0\to
  \left(
  (1-\epsilon)\frac{n}{2}
  \leq
  \lvert uX\rvert
  \leq
  (1+\epsilon)\frac{n}{2}\right)
  \right)\right\}.
\end{align*}

We recall that for an $\epsilon$-balanced code, the MRRW bound on the rate is of the form
\begin{align}\label{eq:MRRWrate}
  \frac{1+o(1)}{4}\epsilon^2\lg\frac{1}{\epsilon} + O_\epsilon\left(\frac{\lg(n)}{n}\right)
\end{align}
as $n\to\infty$ and $\epsilon\to 0$ (in the above, the error term $O_\epsilon(\lg(n)/n)$ hides
multiplicative factors dependent on $\epsilon$, but the error term $o(1)$ only hides multiplicative
factors that do not depend on $n$ nor on $\epsilon$). We will retrieve this bound on every constant
level of the hierarchy. However, we point out right away that the error terms hidden are slightly
worse than the MRRW bound and get worse as the level increases.

Recall that the LP~\eqref{eq:LPdual} is symmetric under the action of $S_n$, and so is the solution
we construct. Namely, it is constant on the orbits $\FF_2^{\ell\times n}/S_n$. As it turns out,
$S_n$-orbits can be characterized in terms of \emph{configurations}, defined below
in~\eqref{eq:config_def}. In \cref{subsec:spectral_defs} we develop the language and tools necessary
to work with symmetric functions.

In \cref{subsec:spectral_functions} we construct a family of 
feasible solutions of the form\footnote{In \cref{sec:overview}, we used the notation
$\Gamma$ which relates to $\Lambda$ by $\Gamma = \widehat{\Lambda}$, up to a positive
multiplicative factor.}
\begin{align*}
  f(X)
  & \df
  \frac{\Phi_m(X)\cdot \widehat{\Lambda}^2(X)}{(\widehat{\Phi}_m*\Lambda*\Lambda)(0)},
\end{align*}
where $\Phi_m$ is non-positive on $X\in\Valid_{n,\ell}^\epsilon$, and
$\Lambda(X)\df\One[\config_{n,\ell}(X)=h]$ for some $h\in\Config_{n,\ell}$.

The definition of $\Phi_m$ is given in~\eqref{eq:Phim}, and its necessary properties in \cref{lem:Mm}. It can
be viewed, informally, as the product of $2^{\ell}-1$ cylinders in $\RR^{\FF_2^\ell\setminus\{0\}}$ (see
\cref{fig:tetrahedron,fig:valid,fig:cylinders,fig:tetrahedronVD,fig:validVD,fig:cylindersVD}). Each cylinder is negative on the inside and positive on the outside. The cylinders are
centered and rotated so that every $X\in\Valid_{n,\ell}^\epsilon$ is inside an odd number of cylinders, and
hence $\Phi_m(X)\leq 0$.

\input{cylinders}

In \cref{theo:spectral_construction} we prove that the construction yields a
feasible solution, given that $\Lambda$ satisfies certain conditions.
The theorem also provides an upper bound on the objective value
attained by this construction,
and hence on $|C|^\ell$ for $C\in\Valid_{n}$.

Finally, in \cref{sec:find_good_configs} we find a satisfactory $\Lambda$
by choosing a configuration $h\in\Config_{n,\ell}$, and showing that it
satisfies \cref{theo:spectral_construction} and gives the correct value.

\subsection{Basic definitions and properties}\label{subsec:spectral_defs}

This section is dedicated to basic definitions and properties working up to \cref{lem:Apower}, which
provides an easier formula for the action of powers of the matrix $A_v$ defined below.

For $X\in\FF_q^{\ell\times n}$, the \emph{(Venn diagram) configuration} of $X$ is the function
$\config_{n,\ell}(X)\colon\FF_q^\ell\to\NN$ given by letting for each $u\in\FF_q^\ell$
\begin{align*}
  \config_{n,\ell}(X)(u)
  & \df
  \lvert\{k\in[n] \mid \forall j\in[\ell], X_{jk} = u_j\}\rvert
\end{align*}
be the number of columns of $X$ that are equal to $u$. It is straightforward to check that two
elements $X$ and $Y$ of $\FF_q^{\ell\times n}$ are in the same $S_n$-orbit if and only if
$\config_{n,\ell}(X)=\config_{n,\ell}(Y)$. The set of all configurations is denoted by
\begin{align}\label{eq:config_def}
  \Config_{n,\ell}\df\config_{n,\ell}(\FF_q^{\ell\times n})
  & =
  \{g\colon\FF_q^\ell\to\NN\mid\sum_{u\in\FF_q^\ell} g(u) = n\}.
\end{align}
It will be convenient to use the set
\begin{align*}
  \NConfig_\ell
  & \df
  \left\{G\colon\FF_2^\ell\to\RR_+ \;\middle\vert\;
  \sum_{v\in\FF_2^\ell} G(v)=1\right\}
\end{align*}
of normalized Venn diagram configurations over $\FF_2$ (note that we can naturally interpret
elements of $\NConfig_\ell$ as probability distributions on $\FF_2^\ell$).

For $h\in\Config_{n,\ell}$, we let $A_h\in\RR^{\FF_2^{\ell\times n}\times\FF_2^{\ell\times n}}$ and
$L_h\in\RR^{\FF_2^{\ell\times n}}$ be given by
\begin{align*}
  A_h(x,y) & \df \One[\config_{n,\ell}(x-y) = h], &
  L_h(x) & \df 2^{n\ell} \One[\config_{n,\ell}(x) = h],
\end{align*}
and note that
\begin{align*}
  A_h \Lambda & = L_h * \Lambda.
\end{align*}

For every $u\in\FF_2^\ell\setminus\{0\}$, define $h_u\in\Config_{n,\ell}$ by
\begin{align*}
  h_u(v)
  & \df
  \begin{dcases*}
    1, & if $u = v$,\\
    n-1, & if $u = 0$,\\
    0, & otherwise,
  \end{dcases*}
\end{align*}
and define the shorthand notations $A_u\df A_{h_u}$ and $L_u\df L_{h_u}$.

\begin{lemma}\label{lem:volume}
  For $\ell,n\in\NN_+$ and $g\in\Config_{n,\ell}$, we have
  \begin{align*}
    \lvert\config_{n,\ell}^{-1}(g)\rvert
    & =
    \binom{n}{g}.
  \end{align*}

  In particular, if $G\in\NConfig_\ell$ is such that $G(u) > 0$ for every $u\in\FF_2^\ell$ and
  $n\cdot G\in\Config_{n,\ell}$, then
  \begin{align*}
    \lvert\config_{n,\ell}^{-1}(n\cdot G)\rvert
    =
    (1+o(1))
    \cdot\sqrt{\frac{(2\pi n)^{(1-2^\ell)}}{\prod_{u\in\FF_2^\ell}G(u)}}
    \cdot 2^{H_2(G)\cdot n}
  \end{align*}
  as $n\to\infty$ with $\ell$ fixed, where $H_2(G)$ is the binary entropy of $G$ (as a probability
  distribution over $\FF_2^\ell$).
\end{lemma}

\begin{proof}
  By definition, every $X\in\config_{n,\ell}^{-1}(g)$ must be such that for every
  $u\in\FF_2^\ell$, exactly $g(u)$ of the $n$ columns of $X$ must be equal to $u$. Thus, we conclude
  that
  \begin{align*}
    \lvert\config_{n,\ell}^{-1}(g)\rvert
    & =
    \binom{n}{g}
    =
    \frac{n!}{\prod_{u\in\FF_2^\ell} g(u)!}.
  \end{align*}

  Finally, if $G\in\NConfig_\ell$ is such that $G(u) > 0$ for every $u\in\FF_2^\ell$ and $n\cdot
  G\in\Config_{n,\ell}$, then
  \begin{align*}
    \lvert\config_{n,\ell}^{-1}(n\cdot G)\rvert
    & =
    \binom{n}{n\cdot G}
    =
    (1+o(1))
    \cdot\sqrt{\frac{(2\pi n)^{(1-2^\ell)}}{\prod_{u\in\FF_2^\ell}G(u)}}
    \cdot\frac{1}{\prod_{u\in\FF_2^\ell} G(u)^{G(u)\cdot n}}
    \\
    & =
    (1+o(1))
    \cdot\sqrt{\frac{(2\pi n)^{(1-2^\ell)}}{\prod_{u\in\FF_2^\ell}G(u)}}
    \cdot 2^{H_2(G)\cdot n},
  \end{align*}
  where the second equality follows from Stirling's Approximation.
\end{proof}

\begin{lemma}\label{lem:cFgh}
  Let $g,h\in\Config_{n,\ell}$ and let
  \begin{align}\label{eq:cFgh}
    \cF_{g,h}
    & \df
    \left\{F\colon \FF_2^\ell\times\FF_2^\ell\to\NN \;\middle\vert\;
    \sum_{u\in\FF_2^\ell} F(u,\place) = g\land
    \sum_{v\in\FF_2^\ell} F(\place,v) = h
    \right\}.
  \end{align}

  Then the following hold for $Y\in\config_{n,\ell}^{-1}(g)$.
  \begin{enumerate}
  \item\label{lem:cFgh:supp} For every $X\in\config_{n,\ell}^{-1}(h)$, let
    $F_X\colon\FF_2^\ell\times\FF_2^\ell\to\NN$ be given by letting
    \begin{align}\label{eq:cFgh:Fx}
      F_X(u,v)
      & \df
      \lvert\{k\in[n] \mid \forall j\in[\ell], (X_{jk} = u_j\land Y_{jk} = v_j)\}\rvert
    \end{align}
    be the number of indices $k\in[n]$ such that the $k$th column of $X$ is $u$ and the $k$th column
    of $Y$ is $v$. Then $F_X\in\cF_{g,h}$.
  \item\label{lem:cFgh:count} For $F\in\cF_{g,h}$, we have
    \begin{align*}
      \lvert\{X\in\config_{n,\ell}^{-1}(h) \mid F_X=F\}\rvert
      & =
      \prod_{v\in\FF_2^\ell}\binom{g(v)}{F(\place,v)},
    \end{align*}
    where $F_X$ is given by~\eqref{eq:cFgh:Fx}.
  \end{enumerate}
\end{lemma}

\begin{proof}
  \Cref{lem:cFgh:supp} follows since for every $v\in\FF_2^\ell$, we have
  \begin{align*}
    \sum_{u\in\FF_2^\ell} F_X(u,v)
    & =
    \lvert\{k\in[n] \mid \forall j\in[\ell], Y_{jk} = v_j\}\rvert
    =
    \config_{n,\ell}(Y)(v)
    =
    g(v)
  \end{align*}
  and for every $u\in\FF_2^\ell$, we have
  \begin{align*}
    \sum_{v\in\FF_2^\ell} F_x(u,v)
    & =
    \lvert\{k\in[n] \mid \forall j\in[\ell], X_{jk} = u_j\}\rvert
    =
    \config_{n,\ell}(X)(u)
    =
    h(u).
  \end{align*}

  For \cref{lem:cFgh:count}, we note that to count the number of
  $X\in\config_{n,\ell}^{-1}(h)$ with $F_X=F$, we consider $[n]$ partitioned naturally into
  $2^\ell$ parts indexed by $v\in\FF_2^\ell$ as
  \begin{align*}
    P_v & \df \{k\in[n] \mid \forall j\in[\ell], Y_{jk} = v_j\}
  \end{align*}
  and note that to get $F_X=F$ for each $u\in\FF_2^\ell$, we must have exactly $F(u,v)$ elements of
  $P_v$ in
  \begin{align*}
    \{k\in[n] \mid \forall j\in[\ell], (X_{jk} = u_j\land Y_{jk} = v_j)\},
  \end{align*}
  since the above are pairwise disjoint and $\lvert P_v\rvert=g(v)$, we conclude that the number of
  such choices amounts to the multinomial
  \begin{align*}
    \binom{g(v)}{F(\place,v)}
  \end{align*}
  (recall that $\sum_{u\in\FF_2^\ell} F(u,v)=g(v)$, so the multinomial above is non-zero). Since all
  such choices are independent for the different $v\in\FF_2^\ell$, we conclude that
  \begin{align*}
    \lvert\{X\in\config_{n,\ell}^{-1}(h) \mid F_X=F\}\rvert
    & =
    \prod_{v\in\FF_2^\ell}\binom{g(v)}{F(\place,v)},
  \end{align*}
  as desired.
\end{proof}

\begin{lemma}\label{lem:Ahpsi}
  Let $\Psi\colon\Config_{n,\ell}\to\RR$, let $\psi\df\Psi\comp\config_{n,\ell}$, let
  $g,h\in\Config_{n,\ell}$ and let $Y\in\config_{n,\ell}^{-1}(g)$. Then
  \begin{align*}
    A_h \psi(Y)
    & =
    \sum_{F\in\cF_{g,h}} \prod_{w\in\FF_2^\ell}\binom{g(w)}{F(\place,w)}
    \Psi(g+\Delta_F),
  \end{align*}
  where $\cF_{g,h}$ is given by~\eqref{eq:cFgh} and
  \begin{align*}
    \Delta_F(v) & \df \sum_{u\in\FF_2^\ell} (F(u,u+v) - F(u,v)).
  \end{align*}
\end{lemma}

\begin{proof}
  First note that
  \begin{align*}
    A_h\psi(Y)
    & =
    \sum_{\substack{Z\in\FF_2^{\ell\times n}\\\config_{n,\ell}(Z-Y)=h}} \psi(Z)
    =
    \sum_{X\in\config_{n,\ell}^{-1}(h)} \psi(X+Y).
  \end{align*}

  We now split the sum above based on the joint configuration of $X$ and $Y$, that is, for
  $X\in\config_{n,\ell}^{-1}(h)$, we let $F_X\colon\FF_2^\ell\times\FF_2^\ell\to\NN$ be given
  by~\eqref{eq:cFgh:Fx}, i.e., we have
  \begin{align*}
    F_X(u,v)
    & \df
    \lvert\{k\in[n] \mid \forall j\in[\ell], (X_{jk} = u_j\land Y_{jk} = v_j)\}\rvert.
  \end{align*}

  Note that sets in the above partition $[n]$ naturally into $2^\ell\times 2^\ell$ parts indexed by
  $(u,v)\in\FF_2^\ell\times\FF_2^\ell$. Recalling that $\config_{n,\ell}(Y)=g$, we note that
  \begin{align*}
    \config_{n,\ell}(X+Y)(v)
    & =
    \lvert\{k\in[n] \mid \forall j\in[\ell], (X+Y)_{jk}=v_j\}\rvert
    \\
    & =
    \sum_{u\in\FF_2^\ell}
    \lvert\{k\in[n] \mid \forall j\in[\ell], (X_{jk}=u_j\land Y_{jk}=u_j+v_j)\}\rvert
    \\
    & =
    g(v) + \Delta_F(v),
  \end{align*}
  where the last equality follows since $\sum_{u\in\FF_2^\ell} F(u,v) = g(v)$.
\end{proof}

\begin{lemma}\label{lem:Apower}
  Let $v\in\FF_2^\ell\setminus\{0\}$ and $g_0\in\Config_{n,\ell}$ be such that for every
  $u\in\FF_2^\ell$, if $g_0(u)\neq 0$, then $g_0(u)\geq\Omega(n)$. Let also
  $\Lambda\df\One_{\config_{n,\ell}^{-1}(g_0)}$ and $X\in\config_{n,\ell}^{-1}(g_0)$.

  Then
  \begin{align*}
    A_v^m \Lambda(X)
    & =
    \sum_{F\in\cF_{m,v}} \binom{m}{F}\prod_{u\in\FF_2^\ell} g_0(u)^{F(u)} + o(n^m),
  \end{align*}
  as $n\to\infty$ with $m$ and $\ell$ fixed, where
  \begin{align}\label{eq:cFmv}
    \cF_{m,v}
    & \df
    \left\{F\colon\FF_2^\ell\to\NN
    \;\middle\vert\;
    \sum_{u\in\FF_2^\ell} F(u) = m\land
    \forall u\in\FF_2^\ell, F(u+v) = F(u)
    \right\}.
  \end{align}
\end{lemma}

\begin{proof}
  Applying \cref{lem:Ahpsi} for the particular case when $h = h_v$, every $F\in\cF_{g,h}$ is of
  the form $F=F_t$ for some $t\in\FF_2^\ell$, where
  \begin{align*}
    F_t(u,w) & \df
    \begin{dcases*}
      1, & if $u = v$ and $w=t$,\\
      g(t)-1, & if $u = 0$ and $w=t$,\\
      g(w), & if $u = 0$ and $w\neq t$,\\
      0, & otherwise.
    \end{dcases*}
  \end{align*}
  Furthermore, note that we have $\Delta_{F_t} = \One_{\{v+t\}} - \One_{\{t\}}$ and
  \begin{align*}
    \prod_{w\in\FF_2^\ell} \binom{g(w)}{F_t(\place,w)} & = g(t).
  \end{align*}

  Thus we have
  \begin{align*}
    A_v \psi(Y) & = \sum_{t\in\FF_2^\ell} g(t)\Psi(g + \One_{\{v+t\}} - \One_{\{t\}})
  \end{align*}
  and with a simple induction, we have
  \begin{align*}
    A_v^m \psi(Y)
    & =
    \sum_{t\in T_m(g)} \left(\prod_{j=1}^m g_{t,j-1}(t_j)\right)\Psi(g_{t,m}),
  \end{align*}
  where
  \begin{align*}
    g_{t,j} & \df g + \sum_{k=1}^{j-1} (\One_{\{v+t_k\}} - \One_{\{t_k\}}),\\
    T_m(g)
    & \df
    \{t\in(\FF_2^\ell)^m \mid \forall j\in[m], g_{t,j}\in\Config_{n,\ell}, g_{t,j-1}(t_j)\neq 0\}.
  \end{align*}

  For our particular case, we have $\psi = \Lambda = \One_{\config_{n,\ell}^{-1}(g_0)}$ and $\Psi =
  \One_{\{g_0\}}$ and since $m$ is constant and for every $u\in\FF_2^\ell$, if $g_0(u)\neq 0$, then
  $g_0(u)\geq\Omega(n)$, it follows that for $n$ sufficiently large, we have $T_m(g_0) =
  (\FF_2^\ell)^m$ and for every $t\in (\FF_2^\ell)^m$, $j\in[m]$ and $u\in\FF_2^\ell$, we have
  $(g_0)_{t,j}(u) = g_0(u) + o(n)$. Thus, since $X\in\config_{n,\ell}^{-1}(g_0)$, we have
  \begin{align*}
    A_v^m \psi(X)
    & =
    \sum_{t\in(\FF_2^\ell)^m} \left(\prod_{j=1}^m g_0(t_j)\right) \Psi((g_0)_{t,m})
    +
    o(n^m),
  \end{align*}
  where the error term follows since both $m$ and $\ell$ are constants. Thus, we get
  \begin{align*}
    A_v^m \Lambda(X)
    & =
    \sum_{t\in T} \prod_{j=1}^m g_0(t_j) + o(n^m),
  \end{align*}
  where
  \begin{align*}
    T
    & \df
    \{t\in(\FF_2^\ell)^m \mid (g_0)_{t,m} = g_0\}.
  \end{align*}

  For each $t\in T$, let us define a function $F_t\colon\FF_2^\ell\to\NN$ by $F_t(u)\df\lvert
  t^{-1}(u)\rvert$. Note that we must have $\sum_{u\in\FF_2^\ell} F_t(u) = m$ and since $(g_0)_{t,m}
  = g_0$, we must have
  \begin{align*}
    \sum_{u\in\FF_2^\ell} F_t(u)(\One_{\{v+u\}} - \One_{\{u\}}) = 0,
  \end{align*}
  which is equivalent to
  \begin{align*}
    \forall u\in\FF_2^\ell, F_t(u+v) = F_t(u).
  \end{align*}

  It is straightforward to check that for $\cF_{m,v}$ as in~\eqref{eq:cFmv}, we have $\{F_t \mid t\in
  T\} = \cF_{m,v}$ and that for each $F\in\cF_{m,v}$, we have
  \begin{align*}
    \lvert\{t\in T \mid F_t = F\}\rvert
    & =
    \binom{m}{F}.
  \end{align*}
  Thus, we get
  \begin{align*}
    A_v^m \Lambda(X)
    & =
    \sum_{F\in\cF_{m,v}} \binom{m}{F}\prod_{u\in\FF_2^\ell} g_0(u)^{F(u)} + o(n^m),
  \end{align*}
  as desired.
\end{proof}

\subsection{The key functions and matrices}\label{subsec:spectral_functions}

In this section, we provide an abstract way of constructing dual solutions
(\cref{theo:spectral_construction}). We refer the reader to \cref{subsec:overview:spectral} for an
informal description.

Given $\ell,n\in\NN_+$ and $\epsilon\in(0,1)$, for every $m\in\NN$ and every
$u\in\FF_2^\ell\setminus\{0\}$, we let
\begin{align*}
  \phi_{m,u}(X)
  & \df
  \sum_{\substack{v\in\FF_2^\ell\\\langle u,v\rangle=1}}
  \bigl((n - 2\lvert vX\rvert)^m - (\epsilon n)^m\bigr),
  \\
  B_{m,u}
  & \df
  \sum_{\substack{v\in\FF_2^\ell\\\langle u,v\rangle=1}}
  (A_v^m - (\epsilon n)^m I),
\end{align*}
where $\langle u,v\rangle\df\sum_{j\in[\ell]} u_j v_j$.

We also define
\begin{align}\label{eq:Phim}
  \Phi_m & \df \prod_{u\in\FF_2^\ell\setminus\{0\}} \phi_{m,u}, &
  M_m & \df \prod_{u\in\FF_2^\ell\setminus\{0\}} B_{m,u}.
\end{align}
Note that these definitions ensure that
\begin{align}\label{eq:hatPhim}
  2^{n\ell}\widehat{\Phi}_m * \Lambda & = M_m \Lambda
\end{align}
for every $\Lambda\colon\FF_2^{\ell\times n}\to\RR$.

\begin{lemma}\label{lem:constraints}
  For every $u\in\FF_2^\ell\setminus\{0\}$, every $X\in\Valid_{n,\ell}^\epsilon$ and every $m$ even
  such that
  \begin{align}\label{lem:constraints:m}
    m
    & \geq
    \frac{\ell-1}{\lg(1/\epsilon)},
  \end{align}
  where $\lg\df\log_2$ is the binary log, the following hold.
  \begin{enumerate}
  \item\label{lem:constraints:nonnegative} If there exists $v\in\FF_2^\ell$ with $\langle u,v\rangle=1$ and
    $vX = 0$, then $\phi_{m,u}(X)\geq 0$.
  \item\label{lem:constraints:nonpositive} If $vX\neq 0$ for every $v\in\FF_2^\ell$ with
    $\langle u,v\rangle=1$, then $\phi_{m,u}(X)\leq 0$.
  \item\label{lem:constraints:final} If $X\neq 0$, then $\Phi_m(X)\leq 0$.
  \item\label{lem:constraints:value} We have
    \begin{align*}
      \Phi_m(0) & = (2^{\ell-1} (1-\epsilon^m) n^m)^{2^\ell-1}.
    \end{align*}
  \end{enumerate}
\end{lemma}

\begin{proof}
  For \cref{lem:constraints:nonnegative}, note that since $m$ is even and $\langle u,v\rangle =
  1$, we have
  \begin{align*}
    \phi_{m,u}(X)
    & =
    \sum_{\substack{v'\in\FF_2^\ell\\\langle u,v'\rangle=1}}
    \bigl((n - 2\lvert v'X\rvert)^m - (\epsilon n)^m\bigr)
    \geq
    (n - 2 \lvert vX\rvert)^m
    - 2^{\ell-1}\cdot (\epsilon n)^m
    \\
    & \geq
    n^m - 2^{\ell-1}\cdot (\epsilon n)^m
    \geq
    0,
  \end{align*}
  where the last inequality follows from~\eqref{lem:constraints:m}.

  \medskip

  For \cref{lem:constraints:nonpositive}, note that since $X\in\Valid_{n,\ell}^\epsilon$ and
  \begin{align*}
    \phi_{m,u}(X)
    & =
    \sum_{\substack{v'\in\FF_2^\ell\\\langle u,v'\rangle=1}}
    \bigl((n - 2\lvert v'X\rvert)^m - (\epsilon n)^m\bigr),
  \end{align*}
  each $n - 2\lvert v'X\rvert$ in the above is between $-\epsilon n$ and $\epsilon n$, so since $m$
  is even, we get $\phi_{m,u}(X)\leq 0$.

  \medskip

  For \cref{lem:constraints:final}, let $V\df\{v\in\FF_2^\ell \mid vX =
  0\}$. Clearly $V$ is a linear subspace of $\FF_2^\ell$ and since $X\neq 0$, we have
  $V\neq\FF_2^\ell$.

  Note now the following chain of equivalences
  \begin{align*}
    u\in V^\bot
    & \iff
    \forall v\in\FF_2^\ell, (vX = 0\to\langle v,u\rangle = 0)
    \iff
    \forall v\in\FF_2^\ell, (\langle v,u\rangle = 1\to vX\neq 0),
  \end{align*}
  so by \cref{lem:constraints:nonpositive}, we get $\phi_{m,u}(X)\leq 0$ for every $u\in
  V^\bot\setminus\{0\}$.

  On the other hand, note that if $u\in\FF_2^\ell\setminus V^\bot$, then the equivalence above
  implies that there exists $v\in\FF_2^\ell$ with $\langle v,u\rangle = 1$ and $vX = 0$, so
  \cref{lem:constraints:nonnegative} implies $\phi_{m,u}(X)\geq 0$.

  Since $V\neq\FF_2^\ell$, we have $V^\bot\neq\{0\}$, so $\lvert V^\bot\setminus\{0\}\rvert$ is odd,
  hence $\Phi_{m,u}(X)\leq 0$ as it is a product of an odd number of non-positive factors and some
  non-negative factors.

  \medskip

  Finally, \cref{lem:constraints:value} follows by direct calculation.
\end{proof}

We now compute an alternative formula for $M_m$.
\begin{lemma}\label{lem:Mm}
  We have
  \begin{align}
    M_m
    & =
    \sum_{\substack{S\subseteq\FF_2^\ell\setminus\{0\}\\\lvert S\rvert\text{ odd}}}
    \sum_{i\in S}
    \left(
    \prod_{u\in S\setminus\{i\}}
    \sum_{\substack{v\in\FF_2^\ell\\\langle u,v\rangle=1}} A_v^m
    \right)
    \cdot
    (\epsilon n)^{m(2^\ell-1-\lvert S\rvert)}
    \left(\frac{1}{\lvert S\rvert}\cdot\sum_{\substack{v\in\FF_2^\ell\\\langle i,v\rangle=1}} A_v^m
    - \frac{2^{\ell-1}\cdot (\epsilon n)^m}{2^\ell-\lvert S\rvert}\right).
    \label{eq:Mm}
  \end{align}
\end{lemma}

\begin{proof}
  Let
  \begin{align*}
    V
    & \df
    \{v\colon\FF_2^\ell\setminus\{0\}\to\FF_2^\ell \mid
    \forall u\in\FF_2^\ell\setminus\{0\}, \langle u,v(u)\rangle = 1\},
    \\
    M_{m,v}
    & \df
    \sum_{\substack{S\subseteq \FF_2^\ell\setminus\{0\}\\\lvert S\rvert\text{ odd}}}
    \sum_{i\in S}
    \left(\prod_{u\in S\setminus\{i\}} A_{v(u)}^m\right)
    (\epsilon n)^{m(2^\ell-1-\lvert S\rvert)}
    \left(\frac{A_{v(i)}^m}{\lvert S\rvert}
    - \frac{(\epsilon n)^m}{2^\ell - \lvert S\rvert}\right) \qquad (v\in V).
  \end{align*}
  We will first show that $M_m = \sum_{v\in V} M_{m,v}$.
  
  Note that
  \begin{align*}
    M_m
    & =
    \prod_{u\in\FF_2^\ell\setminus\{0\}} B_{m,u}
    =
    \prod_{u\in\FF_2^\ell\setminus\{0\}}
    \sum_{\substack{v\in\FF_2^\ell\\\langle u,v\rangle=1}}
    (A_v^m - (\epsilon n)^m I)
    =
    \sum_{v\in V}
    \prod_{u\in\FF_2^\ell\setminus\{0\}} (A_{v(u)}^m - (\epsilon n)^m I).
  \end{align*}
  Our objective is then to show that the inner product in the above is equal to $M_{m,v}$. To prove
  this, note that
  \begin{align*}
    \prod_{u\in\FF_2^\ell\setminus\{0\}} (A_{v(u)}^m - (\epsilon n)^m I)
    & =
    \sum_{S\subseteq \FF_2^\ell\setminus\{0\}}
    \left(\prod_{u\in S} A_{v(u)}^m\right)
    (-(\epsilon n)^m)^{2^\ell-1-\lvert S\rvert}.
  \end{align*}

  We now group the terms in the sum above as follows: we sum over only
  $S\subseteq\FF_2^\ell\setminus\{0\}$ such that $\lvert S\rvert$ is odd and we
  redistribute the terms with $\lvert S\rvert$ even equally among $S\cup\{i\}$ where $i$
  ranges in $\FF_2^\ell\setminus(\{0\}\cup S)$. With this redistribution, we have
  \begin{align*}
    & \!\!\!\!\!\!
    \prod_{u\in\FF_2^\ell\setminus\{0\}} (A_{v(u)}^m - (\epsilon n)^m I)
    \\
    & =
    \sum_{\substack{S\subseteq\FF_2^\ell\setminus\{0\}\\\lvert S\rvert\text{ odd}}}
    \left(
    \left(\prod_{u\in S} A_{v(u)}^m\right)
    (-(\epsilon n))^{m(2^\ell-1-\lvert S\rvert)}
    +
    \sum_{i\in S}
    \frac{1}{2^\ell-\lvert S\rvert}
    \left(\prod_{u\in S\setminus\{i\}} A_{v(u)}^m\right)
    (-(\epsilon n)^m)^{2^\ell-\lvert S\rvert}
    \right)
    \\
    & =
    \sum_{\substack{S\subseteq\FF_2^\ell\setminus\{0\}\\\lvert S\rvert\text{ odd}}}
    \sum_{i\in S}
    \left(\prod_{u\in S\setminus\{i\}} A_{v(u)}^m\right) (\epsilon n)^{m(2^\ell-1-\lvert S\rvert)}
    \left(\frac{A_{v(i)}^m}{\lvert S\rvert}
    - \frac{(\epsilon n)^m}{2^\ell-\lvert S\rvert}\right)
    \\
    & =
    M_{m,v},
  \end{align*}
  so we conclude that $M_m = \sum_{v\in V} M_{m,v}$.

  Finally, note that
  \begin{align*}
    M_m
    & =
    \sum_{v\in V} M_{m,v}
    \\
    & =
    \sum_{v\in V}
    \sum_{\substack{S\subseteq\FF_2^\ell\setminus\{0\}\\\lvert S\rvert\text{ odd}}}
    \sum_{i\in S}
    \left(\prod_{u\in S\setminus\{i\}} A_{v(u)}^m\right)
    (\epsilon n)^{m(2^\ell-1-\lvert S\rvert)}
    \left(\frac{A_{v(i)}^m}{\lvert S\rvert}
    - \frac{(\epsilon n)^m}{2^\ell-\lvert S\rvert}\right)
    \\
    & =
    \sum_{\substack{S\subseteq\FF_2^\ell\setminus\{0\}\\\lvert S\rvert\text{ odd}}}
    \sum_{i\in S}
    \left(
    \prod_{u\in S\setminus\{i\}}
    \sum_{\substack{v\in\FF_2^\ell\\\langle u,v\rangle=1}} A_v^m
    \right)
    \cdot
    (\epsilon n)^{m(2^\ell-1-\lvert S\rvert)}
    \sum_{\substack{v\in\FF_2^\ell\\\langle i,v\rangle=1}}
    \left(\frac{A_v^m}{\lvert S\rvert}
    - \frac{(\epsilon n)^m}{2^\ell-\lvert S\rvert}\right)
    \\
    & =
    \sum_{\substack{S\subseteq\FF_2^\ell\setminus\{0\}\\\lvert S\rvert\text{ odd}}}
    \sum_{i\in S}
    \left(
    \prod_{u\in S\setminus\{i\}}
    \sum_{\substack{v\in\FF_2^\ell\\\langle u,v\rangle=1}} A_v^m
    \right)
    \cdot
    (\epsilon n)^{m(2^\ell-1-\lvert S\rvert)}
    \left(\frac{1}{\lvert S\rvert}\cdot\sum_{\substack{v\in\FF_2^\ell\\\langle i,v\rangle=1}} A_v^m
    - \frac{2^{\ell-1}\cdot (\epsilon n)^m}{2^\ell-\lvert S\rvert}\right),
  \end{align*}
  so~\eqref{eq:Mm} follows.
\end{proof}

\begin{theorem}\label{theo:spectral_construction}
  Let $\ell,m\in\NN_+$ with $m$ even such that
  \begin{align}\label{eq:adhoc:m}
    m
    & \geq
    \frac{\ell-1}{\lg(1/\epsilon)},
  \end{align}
  where $\lg\df\log_2$ is the binary log.

  Suppose further $G\in\NConfig_\ell$ is such that $G(u) > 0$ for every $u\in\FF_2^\ell$.

  Let further $n\in\NN_+$ and suppose that $n\cdot G(u)\in\NN$ for every $u\in\FF_2^\ell$ and that
  for $\Lambda\df\One_{\config_{n,\ell}^{-1}(n\cdot G)}$ and every $i\in\FF_2^\ell\setminus\{0\}$,
  there exists $v\in\FF_2^\ell$ with $\langle i,v\rangle=1$ and
  \begin{align}\label{eq:adhoc:AvmLambda}
    A_v^m \Lambda & \geq (2^{2\ell-1}\epsilon^m n^m + 1)\Lambda.
  \end{align}

  Finally, let
  \begin{align*}
    F & \df \Phi_m\cdot\widehat{\Lambda}^2, &
    f & \df \frac{F}{\widehat{F}(0)},
  \end{align*}
  where $\Phi_m$ is given by~\eqref{eq:Phim}.

  Then $f$ is a feasible solution of~\eqref{eq:LPdual} with
  \begin{align}\label{eq:adhoc:lgf}
    \frac{\lg f(0)}{n} & \leq H_2(G) + O\left(\frac{\lg(n)}{n}\right)
  \end{align}
  as $n\to\infty$ with $\ell$ fixed.
\end{theorem}

\begin{proof}
  It is clear that $\widehat{f}(0)=1$.

  On the other hand, if $X\in\Valid_{n,\ell}^\epsilon\setminus\{0\}$, then by
  \cref{lem:constraints}, we have $\Phi_m(X)\leq 0$, so we get $f(X)\leq 0$.

  For the Fourier constraints, by~\eqref{eq:hatPhim}, we have
  \begin{align*}
    \widehat{f}
    & =
    \frac{\widehat{\Phi_m} * \Lambda * \Lambda}{\widehat{F}(0)}
    =
    \frac{M_m\Lambda * \Lambda}{2^{n\ell}\widehat{F}(0)}.
  \end{align*}
  Since $\Lambda\geq 0$, to show that $\widehat{f}\geq 0$, it suffices to show that $M_m\Lambda\geq
  0$.

  By \cref{lem:Mm}, we have
  \begin{align*}
    M_m
    & =
    \sum_{\substack{S\subseteq\FF_2^\ell\setminus\{0\}\\\lvert S\rvert\text{ odd}}}
    \sum_{i\in S}
    \left(
    \prod_{u\in S\setminus\{i\}}
    \sum_{\substack{v\in\FF_2^\ell\\\langle u,v\rangle=1}} A_v^m
    \right)
    \cdot
    (\epsilon n)^{m(2^\ell-1-\lvert S\rvert)}
    \left(\frac{1}{\lvert S\rvert}\cdot\sum_{\substack{v\in\FF_2^\ell\\\langle i,v\rangle=1}} A_v^m
    - \frac{2^{\ell-1}\cdot (\epsilon n)^m}{2^\ell-\lvert S\rvert}\right)
  \end{align*}
  and from the factoring above, it suffices to show that for every
  $S\subseteq\FF_2^\ell\setminus\{0\}$ with $\lvert S\rvert$ odd and every $i\in S$, we have
  \begin{align*}
    \frac{1}{\lvert S\rvert}\cdot\sum_{\substack{v\in\FF_2^\ell\\\langle i,v\rangle=1}} A_v^m
    \Lambda
    & \geq
    \frac{2^{\ell-1}\cdot (\epsilon n)^m}{2^\ell-\lvert S\rvert} \Lambda.
  \end{align*}

  Since $1\leq\lvert S\rvert\leq 2^\ell-1$, it suffices to then show that
  \begin{align*}
    \frac{1}{2^\ell}\cdot\sum_{\substack{v\in\FF_2^\ell\\\langle i,v\rangle=1}} A_v^m
    \Lambda
    & \geq
    2^{\ell-1}\cdot (\epsilon n)^m \Lambda,
  \end{align*}
  which follows directly from our assumption~\eqref{eq:adhoc:AvmLambda} (and the fact that all entries of
  $A_v^m$ and $\Lambda$ are non-negative). Note that since we have an extra $1$
  in~\eqref{eq:adhoc:AvmLambda}, the argument above in fact implies
  \begin{align}\label{eq:MmLambda}
    M_m\Lambda & \geq \poly(n)\Lambda.
  \end{align}

  It remains to show~\eqref{eq:adhoc:lgf}. By \cref{lem:volume,lem:constraints:value}, we have
  \begin{align*}
    F(0)
    & =
    \Phi_m(0)\cdot\widehat{\Lambda}(0)^2
    =
    (2^{\ell-1} (1-\epsilon^m) n^m)^{2^\ell-1}
    \cdot
    \left(
    \frac{\lvert\config_{n,\ell}^{-1}(n\cdot G)\rvert}{2^{n\ell}}
    \right)^2
    =
    \poly(n)\cdot 2^{2(H_2(G) - \ell)n}.
  \end{align*}

  On the other hand, we have
  \begin{align*}
    \widehat{F}(0)
    & =
    (\widehat{\Phi_m} * \Lambda * \Lambda)(0)
    =
    \frac{(M_m\Lambda * \Lambda)(0)}{2^{n\ell}}
    \geq
    \frac{\poly(n)}{2^{n\ell}}(\Lambda * \Lambda)(0)
    =
    \poly(n)\cdot 2^{(H_2(G) - 2\ell)n},
  \end{align*}
  where the inequality follows from~\eqref{eq:MmLambda} and the last equality follows from
  \cref{lem:volume}. Thus, we get
  \begin{align*}
    \frac{\lg(f(0))}{n}
    & \leq
    H_2(G) + O\left(\frac{\lg(n)}{n}\right),
  \end{align*}
  as desired.
\end{proof}

\subsection{Finding Good Configurations}
\label{sec:find_good_configs}

\Cref{theo:spectral_construction} leaves open only one question: which normalized configurations $G$
are such that the corresponding function $\Lambda$ satisfies~\eqref{eq:adhoc:AvmLambda} while having
small binary entropy $H_2(G)$ so as to yield a good value to~\eqref{eq:LPdual}? In this section, we
will see that two kinds of normalized configurations can attain same rates as MRRW
(see~\eqref{eq:MRRWrate}) up to lower order terms via \cref{theo:spectral_construction}.

\begin{definition}
  Given $\ell\in\NN_+$ and $\tau\in[0,1/\ell]$, the \emph{$\tau$-vertex uniform normalized
    configuration} (at level $\ell$) is defined as $\Gvu\in\NConfig_\ell$ given by
  \begin{align*}
    \Gvu(u) & \df
    \begin{dcases*}
      (1-\ell\tau), & if $u=0$,\\
      \tau, & if $\lvert u\rvert=1$,\\
      0, & otherwise.
    \end{dcases*}
  \end{align*}

  Given $\tau\in[0,1]$, the \emph{$\tau$-quasirandom normalized configuration} (at level $\ell$) is
  defined as $\GQR\in\NConfig_\ell$ given by
  \begin{align*}
    \GQR(u) & \df \tau^{\lvert u\rvert}(1-\tau)^{\ell-\lvert u\rvert}.
  \end{align*}

  Given further $n\in\NN_+$, we let $\gvu,\gQR$ be obtained by rounding $n\cdot\Gvu$
  and $n\cdot\GQR$ respectively to integer values so that the result is in $\Config_{n,\ell}$.
\end{definition}

\begin{lemma}\label{lem:gvu}
  Let $\epsilon\in(0,1)$, let $\ell\in\NN_+$, let $\tau\in(0,1/\ell)$, let $n,m\in\NN_+$ with $m$
  even and let $\Lambda\df\One_{\config_{n,\ell}^{-1}(\gvu)}$. Then the following hold:
  \begin{enumerate}
  \item\label{lem:gvu:AvmLambda} For every $v\in\FF_2^\ell$ with $\lvert v\rvert=1$ and every
    $X\in\config_{n,\ell}^{-1}(\gvu)$, we have
    \begin{align*}
      A_v^m \Lambda(X)
      & = 
      \binom{m}{m/2} (1-\ell\tau)^{m/2}\tau^{m/2} n^m + o(n^m).
    \end{align*}
  \item\label{lem:gvu:entropy} We have
    \begin{align*}
      H_2(\Gvu)
      & =
      \ell\left(\tau\lg\frac{1}{\tau} +
      (1-\ell\tau)\lg\frac{1}{1-\ell\tau}\right)
      =
      \ell\tau\lg(\tau) + \ell\tau + O(\tau^2),
    \end{align*}
    as $\tau\to 0$ with $\ell$ fixed.
  \item\label{lem:gvu:walks} If
    \begin{align}\label{eq:gvu:tau}
      \tau
      & =
      \frac{1 - \sqrt{1 - \ell 2^{(4\ell-1)/m} m^{1/m}\epsilon^2}}{2\ell},
    \end{align}
    then
    \begin{align}\label{eq:gvu:tauasymp}
      \tau
      =
      \frac{2^{(4\ell-1)/m} m^{1/m}}{4} \epsilon^2 + O(\epsilon^4)
    \end{align}
    as $\epsilon\to 0$ with $\ell$ and $m$ fixed and
    \begin{align}\label{eq:gvu:AvmLambdatau}
      A_v^m \Lambda & \geq 2^{2\ell-1}\epsilon^m n^m \Lambda + o(n^m)
    \end{align}
    for every $v\in\FF_2^\ell$ with $\lvert v\rvert=1$ as $n\to\infty$ with $\epsilon$, $\ell$ and
    $m$ fixed.
  \end{enumerate}
\end{lemma}

\begin{proof}
  First note that since $0 < \tau < 1/\ell$, it follows that for every $u\in\FF_2^\ell$, if
  $\gvu(u)\neq 0$, then $\gvu(u)\geq\Omega(n)$, so by \cref{lem:Apower} with $g_0\df\gvu$, we have
  \begin{align*}
    A_v^m \Lambda(X)
    & =
    \sum_{F\in\cF_{m,v}} \binom{m}{F}\prod_{u\in\FF_2^\ell} \gvu(u)^{F(u)} + o(n^m),
  \end{align*}
  where $\cF_{m,v}$ is given by~\eqref{eq:cFmv}.

  Since $\gvu(u) = 0$ whenever $\lvert u\rvert\geq 2$, it follows that the only terms of the sum
  above that are non-zero correspond to $F\in\cF_{m,v}$ that are entirely supported on
  $\{u\in\FF_2^\ell \mid \lvert u\rvert\leq 1\}$. Since all $F\in\cF_{m,v}$ further satisfy
  $F(u)=F(u+v)$ for every $u\in\FF_2^\ell$, we conclude that only one term of the sum above can be
  non-zero, namely the one corresponding to $F_0\in\cF_{m,v}$ given by
  \begin{align*}
    F_0(u) & \df
    \begin{dcases*}
      \frac{m}{2}, & if $u=0$ or $u=v$,\\
      0, & otherwise,
    \end{dcases*}
  \end{align*}
  so we get
  \begin{align*}
    A_v^m \Lambda(X)
    & =
    \binom{m}{m/2} (1-\ell\tau)^{m/2}\tau^{m/2} n^m + o(n^m),
  \end{align*}
  so \cref{lem:gvu:AvmLambda} holds.

  \medskip

  For \cref{lem:gvu:entropy}, note that
  \begin{align*}
    H_2(\Gvu)
    & =
    \sum_{u\in\FF_2^\ell}\Gvu(u)\lg\frac{1}{\Gvu(u)}
    =
    \ell\tau\lg\frac{1}{\tau} + (1-\ell\tau)\lg\frac{1}{1-\ell\tau}
    \\
    & =
    \ell\tau\lg\frac{1}{\tau} + (1-\ell\tau)(\ell\tau + O(\tau^2))
    =
    \ell\tau\lg\frac{1}{\tau} + \ell\tau + O(\tau^2).
  \end{align*}

  \medskip

  For \cref{lem:gvu:walks}, first note that~\eqref{eq:gvu:tauasymp} follows
  from~\eqref{eq:gvu:tau} and the fact that $\sqrt{1+t} = 1 + t/2 + O(t^2)$ as $t\to 0$.

  Finally, note that~\eqref{eq:gvu:AvmLambdatau} is trivial when evaluated on a point $X$ not in the
  support of $\Lambda$ as the left-hand side is clearly non-negative.

  On the other hand, for $X\in\supp(\Lambda)$, that is, for $X\in\config_{n,\ell}^{-1}(\gvu)$, by
  \cref{lem:gvu:AvmLambda}, we have
  \begin{align*}
    A_v^m \Lambda(X)
    & =
    \binom{m}{m/2} (1-\ell\tau)^{m/2}\tau^{m/2} n^m + o(n^m)
    \\
    & \geq
    \frac{2^m}{\sqrt{2m}}
    \cdot\left(\frac{1 - (1-\ell 2^{(4\ell-1)/m} m^{1/m} \epsilon^2)}{4\ell}\right)^{m/2}
    + o(n^m)
    \\
    & =
    2^{2\ell-1}\epsilon^m + o(n^m),
  \end{align*}
  as desired.
\end{proof}

\begin{lemma}\label{lem:gQR}
  Let $\epsilon\in(0,1)$, let $\ell\in\NN_+$, let $\tau\in(0,1)$, let $n,m\in\NN_+$ with $m$ even
  and let $\Lambda\df\One_{\config_{n,\ell}^{-1}(\gQR)}$. Then the following hold:
  \begin{enumerate}
  \item\label{lem:gQR:AvmLambda} For every $v\in\FF_2^\ell$ with $\lvert v\rvert=1$ and every
    $X\in\config_{n,\ell}^{-1}(\gQR)$, we have
    \begin{align*}
      A_v^m \Lambda(X)
      & =
      \binom{m}{m/2}
      \tau^{m/2}(1-\tau)^{\ell m/2}(1-2\tau+2\tau^2)^{(\ell-1)m/2}
      n^m
      + o(n^m).
    \end{align*}
  \item\label{lem:gQR:entropy} We have
    \begin{align*}
      H_2(\GQR)
      & =
      \ell\left(\tau\lg\frac{1}{\tau} + (1-\tau)\lg\frac{1}{1-\tau}\right)
      =
      \ell\tau\lg\frac{1}{\tau} + \ell\tau + O(\tau^2),
    \end{align*}
    as $\tau\to 0$ with $\ell$ fixed.
  \item\label{lem:gQR:walks} If $\tau$ is the first non-negative root of
    \begin{align}\label{eq:gQR:tau}
      4\tau(1-\tau)^\ell(1-2\tau+2\tau^2)^{\ell-1} - 2^{(4\ell-1)/m} m^{1/m}\epsilon^2
    \end{align}
    then
    \begin{align}\label{eq:gQR:tauasymp}
      \tau
      & =
      \frac{2^{(4\ell-1)/m} m^{1/m}}{4} \epsilon^2 + O(\epsilon^{2(1+\ell)})
    \end{align}
    as $\epsilon\to 0$ with $\ell$ and $m$ fixed and
    \begin{align}\label{eq:gQR:AvmLambdatau}
      A_v^m \Lambda & \geq 2^{2\ell-1}\epsilon^m n^m \Lambda + o(n^m)
    \end{align}
    for every $v\in\FF_2^\ell$ with $\lvert v\rvert=1$ as $n\to\infty$ with $\epsilon$, $\ell$ and
    $m$ fixed.
  \end{enumerate}
\end{lemma}

\begin{proof}
  First note that since $0 < \tau < 1$, it follows that $\gQR(u)\geq\Omega(n)$ for every
  $u\in\FF_2^\ell$, so by \cref{lem:Apower} with $g_0\df\gQR$, we have
  \begin{align*}
    A_v^m \Lambda(X)
    & =
    \sum_{F\in\cF_{m,v}} \binom{m}{F}\prod_{u\in\FF_2^\ell} \gQR(u)^{F(u)} + o(n^m),
  \end{align*}
  where $\cF_{m,v}$ is given by~\eqref{eq:cFmv}.

  Let $i_0\in\supp(v)$ and note that there is a natural one-to-one correspondence between
  $\cF_{m,v}$ and the set
  \begin{align*}
    \cF
    & \df
    \left\{F\colon\FF_2^{[\ell]\setminus\{i_0\}}\to\NN \;\middle\vert\;
    \sum_{u\in\FF_2^{[\ell]\setminus\{i_0\}}} F(u) = \frac{m}{2}\right\}
  \end{align*}
  in which $F\in\cF_{m,v}$ corresponds to $F\rest_{\FF_2^{[\ell]\setminus\{i_0\}}}$. Thus, we get
  \begin{align*}
    A_v^m \Lambda(X)
    & =
    \sum_{F\in\cF} \binom{m}{F,F}\prod_{u\in\FF_2^{[\ell]\setminus\{i_0\}}}
    (\gQR(u)\gQR(u+v))^{F(u)} + o(n^m)
    \\
    & =
    \binom{m}{m/2}
    \sum_{F\in\cF}\binom{m/2}{F}^2\prod_{u\in\FF_2^{[\ell]\setminus\{i_0\}}}
    (\tau^{2\lvert u\rvert+1}(1-\tau)^{2\ell-2\lvert u\rvert-1})^{F(u)}
    n^m
    + o(n^m)
    \\
    & \geq
    \binom{m}{m/2}
    \sum_{F\in\cF}\binom{m/2}{F}\prod_{u\in\FF_2^{[\ell]\setminus\{i_0\}}}
    (\tau^{2\lvert u\rvert+1}(1-\tau)^{2\ell-2\lvert u\rvert-1})^{F(u)}
    n^m
    + o(n^m)
    \\
    & =
    \binom{m}{m/2}
    \left(\sum_{u\in\FF_2^{[\ell]\setminus\{i_0\}}}
    \tau^{2\lvert u\rvert+1}(1-\tau)^{2\ell-2\lvert u\rvert-1}
    \right)^{m/2}
    n^m
    + o(n^m)
    \\
    & =
    \binom{m}{m/2}
    \tau^{m/2}(1-\tau)^{(2\ell-1)m/2}
    \left(1 + \left(\frac{\tau}{1-\tau}\right)^2\right)^{(\ell-1)m/2}
    n^m
    + o(n^m)
    \\
    & =
    \binom{m}{m/2}
    \tau^{m/2}(1-\tau)^{\ell m/2}(1-2\tau+2\tau^2)^{(\ell-1)m/2}
    n^m
    + o(n^m),
  \end{align*}
  where the third equality follows from the Multinomial Theorem and the fourth equality follows from
  the Binomial Theorem. Thus, \cref{lem:gQR:AvmLambda} holds.

  \medskip

  For \cref{lem:gQR:entropy}, note that
  \begin{align*}
    H_2(\GQR)
    & =
    \sum_{u\in\FF_2^\ell}
    \tau^{\lvert u\rvert}(1-\tau)^{\ell-\lvert u\rvert}
    \lg\frac{1}{\tau^{\lvert u\rvert}(1-\tau)^{\ell-\lvert u\rvert}}
    \\
    & =
    \sum_{j=0}^\ell\binom{\ell}{j}
    \tau^j(1-\tau)^{\ell-j}
    \left(j\cdot\lg\frac{1}{\tau} + (\ell-j)\lg\frac{1}{1-\tau}\right)
    \\
    & =
    \ell\left(\tau\lg\frac{1}{\tau} + (1-\tau)\lg\frac{1}{1-\tau}\right)
    \\
    & =
    \ell\tau\lg\frac{1}{\tau} + \ell\tau + O(\tau^2).
  \end{align*}

  \medskip

  For \cref{lem:gQR:walks}, first note that the expression in~\eqref{eq:gQR:tau} takes a
  negative value when $\tau = 0$ and takes the value
  \begin{align*}
    2^{3-2\ell} - 2^{(4\ell-1)/m} m^{1/m}\epsilon^2
  \end{align*}
  when $\tau=1/2$, which is positive if $\epsilon > 0$ is small enough, so the expression
  in~\eqref{eq:gQR:tau} has a non-negative root before $1/2$. If $\tau$ is the first non-negative
  root in~\eqref{eq:gQR:tau}, then~\eqref{eq:gQR:tauasymp} follows straightforwardly.

  Finally, note that~\eqref{eq:gQR:AvmLambdatau} is trivial when evaluated on a point $X$ not in the
  support of $\Lambda$ as the left-hand side is clearly non-negative.

  On the other hand, for $X\in\supp(\Lambda)$, that is, for $X\in(\config_{n,\ell}^v)^{-1}(\gQR)$, by
  \cref{lem:gQR:AvmLambda}, we have
  \begin{align*}
    A_v^m \Lambda(X)
    & =
    \binom{m}{m/2}
    \tau^{m/2}(1-\tau)^{\ell m/2}(1-2\tau+2\tau^2)^{(\ell-1)m/2}
    n^m
    + o(n^m)
    \\
    & \geq
    \frac{2^m}{\sqrt{2m}}
    \tau^{m/2}(1-\tau)^{\ell m/2}(1-2\tau+2\tau^2)^{(\ell-1)m/2}
    n^m
    + o(n^m)
    \\
    & =
    2^{2\ell-1} m^{1/m}\epsilon^m n^m + o(n^m),
  \end{align*}
  where the second equality follows since $\tau$ is a root of~\eqref{eq:gQR:tau}.
\end{proof}

\begin{corollary}\label{cor:spectral_construction}
  Let $\epsilon\in(0,1)$, let $\ell,m\in\NN_+$ with $m$ even such that
  \begin{align*}
    m
    & \geq
    \frac{\ell-1}{\lg(1/\epsilon)},
  \end{align*}
  where $\lg\df\log_2$ is the binary log.

  Then for every sufficiently large $n$, there exist $g_1,g_2\in\Config_{n,\ell}$ with
  \begin{align*}
    \lvert g_1(u) - n\cdot\Gvu(u)\rvert & \leq o(n), &
    \lvert g_2(u) - n\cdot\GQR(u)\rvert & \leq o(n)
  \end{align*}
  for every $u\in\FF_2^\ell$ such that for
  \begin{align*}
    \Lambda_i & \df \One_{\config_{n,\ell}^{-1}(g_i)}, &
    F_i & \df \Phi_m\cdot\widehat{\Lambda}_i^2, &
    f_i & \df \frac{F_i}{\widehat{F}_i(0)},
  \end{align*}
  where $\Phi_m$ is given by~\eqref{eq:Phim}, we have that $f_1$ and $f_2$ are feasible solutions
  of~\eqref{eq:LPdual} with
  \begin{align*}
    \frac{\lg f_i(0)}{n}
    & \leq
    \frac{2^{(4\ell-1)/m} m^{1/m}}{4} \epsilon^2\lg\frac{1}{\epsilon} + O(\epsilon^4) +
    O_\epsilon\left(\frac{\lg(n)}{n}\right)
    \\
    & =
    \frac{1 + o(1)}{4} \epsilon^2\lg\frac{1}{\epsilon}
    + O_\epsilon\left(\frac{\lg(n)}{n}\right)
  \end{align*}
  as $n\to\infty$ and $\epsilon\to 0$ with $\ell$ and $m$ fixed (in the above, the error term
  $O_\epsilon(\lg(n)/n)$ hides multiplicative factors dependent on $\epsilon$, but the error terms
  $o(1)$ and $O(\epsilon^4)$ only hide multiplicative factors that do not depend on $n$ nor on
  $\epsilon$).
\end{corollary}

\begin{proof}
  Follows by combining \cref{lem:gvu,lem:gQR} with \cref{theo:spectral_construction}
  (note that the small adjustment to the configurations is needed both due to the error terms
  in~\eqref{eq:gvu:AvmLambdatau} and~\eqref{eq:gQR:AvmLambdatau} and to obtain the extra $1$ term
  needed in~\eqref{eq:adhoc:AvmLambda}).
\end{proof}

\section{Conclusion}\label{sec:conclusion}

Establishing tight bounds on the rate-vs-distance trade-off of binary codes has remained a major open question in coding theory.
The best existential constructions given by the Gilbert--Varshamov bound have not been improved for over 70 years, and
the best upper bounds given by MRRW bound have not been improved for almost 50 years. These known bounds are the same
even for the important class of linear codes. With the inception of complete linear programming hierarchies for linear
codes extending Delsarte's LPs, an ambitious research program of analyzing these higher-order Delsarte LPs is launched.
On one hand their similarity with the original Delsarte LPs gives hope this might be a viable task. On the other hand,
the higher-order structure poses non-trivial challenges.

We view the contributions of this work as establishing important milestones in this research program as we are able
to construct higher-order dual feasible solutions for the first time. This is done in two complementary ways. First,
by explicitly lifting dual solutions from lower levels to higher levels of these hierarchies. Second, by constructing
higher-order dual solutions from scratch generalizing spectral-based techniques. Given that these constructions
either match or approximately match the best known bounds, together with the proven strength of these complete hierarchies,
they open up important avenues of further exploration. For instance, very interesting concrete questions made possible by
this work are the following.
\begin{itemize}
  \item After lifting a dual solution of the original Delsate LP to a higher-level $\ell$ of these hierarchies, can we improve its objective value and
        improve over the MRRW bound?
 \item We saw that the spectral-based construction has some degrees of freedom, namely, there is a choice of function $\phi$ capturing the sign of
       the valid region and a choice of configurations for an eigenvalue-like problem. Can we find suitable choices to improve the MRRW bound?
\end{itemize}

\section*{Acknowledgments}

The authors are very thankful for the support and hospitality of IAS and the Simons Institute.
LC, FGJ and EL thank Avi Wigderson for hosting them at the wonderful IAS, where part of this work was done.
In particular, we would like to highlight the importance to us of the programs and clusters:
``HDX and Codes'', ``Analysis and TCS: New Frontiers'', ``Error-Correcting Codes: Theory and Practice'',
and ``Quantum Algorithms, Complexity, and Fault Tolerance''.
FGJ thanks Venkat Guruswami for kindly hosting him in his fantastic research groups.
FGJ thanks the support in part as a Google Research Fellow.

CJ is a member of the Bocconi Institute for Data Science and Analytics (BIDSA).
Work supported in part by the European Research Council (ERC) under the European
Union’s Horizon 2020 research and innovation programme (grant agreement Nos. 834861 and 101019547).

\bibliographystyle{alphaurl}
\bibliography{refs}

\appendix

\clearpage

\section{Other Formulations of the Hierarchy}\label{sec:otherform}

In this section we state other formulations that are not used in the current work.

\subsection{\Lovasz\ $\vartheta'$ Formulation}
\label{subsec:Lovasz}

The $\vartheta'$ formulation mentioned in \cref{sec:hierarchies} is~\eqref{eq:SDPprimal}, whose dual
is~\eqref{eq:SDPdual}; a linear code $C\in\Valid_n$ yields a natural solution $M_C$
of~\eqref{eq:SDPprimal} given by $M_C(X,Y) \df \One[X_1,\ldots,X_\ell,Y_1,\ldots,Y_\ell\in C]/\lvert
C\rvert^\ell$, whose value is $\lvert C\rvert^\ell$.

\begin{empheq}[box=\fbox]{equation}\label{eq:SDPprimal}
  \begin{aligned}
    \text{Variables: }
    & \mathrlap{M\colon\FF_q^{\ell\times n}\times\FF_q^{\ell\times n}\to\RR\text{ symmetric}}
    \\
    \max \qquad
    & \sum_{\mathclap{X,Y\in\FF_q^{\ell\times n}}} M(X,Y)
    \\
    \text{s.t.} \qquad
    & \tr(M) = 1
    & &
    & (\text{Normalization})
    \\
    & M(X,Y) = 0
    & & \forall X,Y\in\FF_q^{\ell\times n} \text{ with } X-Y\notin\Valid_{n,\ell}
    & (\text{Validity})
    \\
    & M \succeq 0
    & &
    & \mathllap{(\text{Positive semidefiniteness})}
    \\
    & M(X,Y) \geq 0
    & & \forall X,Y\in\FF_q^{\ell\times n}
    & (\text{Non-negativity})
  \end{aligned}
\end{empheq}

\begin{empheq}[box=\fbox]{equation}\label{eq:SDPdual}
  \begin{aligned}
    \text{Variables: }
    & \mathrlap{N\colon\FF_q^{\ell\times n}\times\FF_q^{\ell\times n}\to\RR\text{ symmetric}}
    \\
    \min\qquad
    & \beta
    \\
    \text{s.t.} \qquad
    & \beta I - N \succeq 0
    & &
    & (\text{Maximum eigenvalue})
    \\
    & N(X,Y) \geq 1
    & & \forall X,Y\in\FF_q^{\ell\times n} \text{ with } X-Y\in\Valid_{n,\ell}
    & (\text{Validity})
  \end{aligned}
\end{empheq}

\subsection{LP Formulation}
\label{subsec:LP}

To get from the $\vartheta'$ formulation of~\eqref{eq:SDPprimal} to the LP formulation
of~\eqref{eq:LPprimal}, one first notes that all $\FF_q^n$-symmetric solutions must lie in the span
of the matrices
\begin{align*}
  E_Z(X,Y) & \df \One[X-Y = Z] & (X,Y,Z\in\FF_q^{\ell\times n}).
\end{align*}
On the other hand, the space of $\FF_q^n$-invariant solutions is the also the span of the Fourier
matrices
\begin{align*}
  F_Z(X,Y) & \df \chi_Z(X)\overline{\chi}_Z(Y)
  &
  (X,Y,Z\in\FF_q^{\ell\times n}),
\end{align*}
which are positive semidefinite. The corresponding change of variables is summarized by
\begin{align*}
  \sum_{Z\in\FF_q^{\ell\times n}} f(Z) E_Z & = \sum_{Z\in\FF_q^{\ell\times n}} \widehat{f}(Z) F_Z, &
  \widehat{f}(Z) & \df \frac{1}{q^{n\ell}}\sum_{X\in\FF_q^{\ell\times n}} f(X)\overline{\chi_Z(X)}
  & (f\in\CC^{\FF_q^{\ell\times n}}).
\end{align*}
Since any $\FF_q^n$-symmetric solution is of the first form above for some $f\colon\FF_q^{\ell\times
  n}\to\CC$, the semidefinite constraint amounts to non-negativity of $\widehat{f}$ and all other
constraints translate easily to linear constraints on $f$.

\subsection{Krawtchouk Formulation}
\label{subsec:kraw}

The Krawtchouk formulation mentioned in \cref{sec:hierarchies} uses the $S_n$-symmetry to
rewrite the Fourier transform in terms of the higher-order Krawtchouk polynomials
$K_h\colon\Config_{n,\ell}\to\CC$ ($h\in\Config_{n,\ell}$) given by
\begin{gather*}
  K_h(g)
  \df
  \sum_{F\in\cF_{g,h}}
  \prod_{w\in\FF_q^\ell}\binom{g(w)}{F(\place,w)}
  \prod_{u,w\in\FF_q^\ell}\chi_u(w)^{F(u,w)},
  \\
  \cF_{g,h}
  \df
  \left\{F\colon\FF_q^\ell\times\FF_q^\ell\to\NN \;\middle\vert\;
  \sum_{u\in\FF_q^\ell} F(u,\place) = g\land
  \sum_{w\in\FF_q^\ell} F(\place,w) = h
  \right\},
  \\
  \begin{aligned}
  \binom{g(w)}{F(\place,w)}
  & \df
  \frac{g(w)!}{\prod_{u\in\FF_q^\ell} F(u,w)!},
  &
  \qquad
  \chi_u(w)
  & \df
  \exp\left(\frac{2\pi i u w}{q}\right).
  \end{aligned}
\end{gather*}

In both the Krawtchouk formulation of~\eqref{eq:KLPprimal} and its dual in~\eqref{eq:KLPdual} below,
$g^-\in\Config_{n,\ell}$ denotes the configuration given by $g^-(u)\df g(-u)$; a linear code
$C\in\Valid_n$ yields a natural solution $f_C$ of~\eqref{eq:KLPprimal} given by $f_C(g) \df
\lvert\{X\in(\config_{n,\ell})^{-1}(g) \mid X_1,\ldots,X_\ell\in C\}\rvert$.

\begin{empheq}[box=\fbox]{equation}\label{eq:KLPprimal}
  \begin{aligned}
    \text{Variables: }
    & \mathrlap{f\colon\Config_{n,\ell}\to\RR}
    \\
    \max \qquad
    & \sum_{g\in\Config_{n,\ell}} f(g)
    \\
    \text{s.t.} \qquad
    & f(\config_{n,\ell}(0)) = 1
    & &
    & (\text{Normalization})
    \\
    & f(g) = 0
    & & \forall g\in\config_{n,\ell}(\FF_q^{\ell\times n}\setminus\Valid_{n,\ell})
    & (\text{Validity})
    \\
    & \sum_{g\in\Config_{n,\ell}} K_h(g) f(g) \geq 0
    & & \forall h\in\Config_{n,\ell}
    & (\text{Krawtchouk})
    \\
    & f(g) \geq 0
    & & \forall g\in\Config_{n,\ell}
    & (\text{Non-negativity})
    \\
    & f(g) = f(g^-)
    & & \forall g\in\Config_{n,\ell}
    & (\text{Symmetry})
  \end{aligned}
\end{empheq}

\begin{empheq}[box=\fbox]{gather}\label{eq:KLPdual}
  \begin{aligned}
    \text{Variables: }
    & \mathrlap{f\colon\Config_{n,\ell}\to\RR, \beta\colon\Config_{n,\ell}\to\RR}
    \\
    \min \quad
    & 1 + \sum_{\mathclap{g\in\Config_{n,\ell}}} K_g(0) f(g)
    \\
    \text{s.t.} \quad
    & 1 + \sum_{\mathclap{g\in\Config_{n,\ell}}} K_g(h) f(g) + \beta(g) - \beta(g^-)\leq 0
    & & \forall h\in\config_{n,\ell}(\Valid_{n,\ell}\setminus\{0\})
    & (\text{Validity})
    \\
    & f(g) \geq 0
    & & \forall g\in\Config_{n,\ell}
    & \mathllap{(\text{Non-negativity})}
  \end{aligned}
\end{empheq}

An alternative way of obtaining~\eqref{eq:KLPprimal} is directly from the \Lovasz\ $\vartheta'$
formulation~\eqref{eq:SDPprimal} by symmetrizing the action of the natural semidirect product
$\FF_q^n\rtimes S_n$ that joins the actions of $\FF_q^n$ and $S_n$ into a single action. In turn,
this amounts to the observation that this $\FF_q^n\rtimes S_n$-action turns $\FF_q^{\ell\times n}$
naturally into a association scheme that is both a translation scheme and Schurian.

\clearpage

\section{Notation}\label{sec:notation}

The set of non-negative integers is denoted by $\NN$ and the set of positive integers is denoted by
$\NN_+\df\NN\setminus\{0\}$. For $n\in\NN$, we let $[n]\df\{1,\ldots,n\}$. We also let $\RR_+$ be
the set of non-negative reals.

For $q,n\in\NN$, we denote the \emph{$n$th geometric sum of ratio $q$} by
\begin{align*}
  [n]_q
  & \df
  \sum_{j=0}^{n-1} q^j
  =
  \begin{dcases*}
    \frac{q^n-1}{q-1}, & if $q\neq 1$,\\
    n, & if $q = 1$.
  \end{dcases*}
\end{align*}
We extend the notation above to when $n\leq 0$ in the natural way so that $\sum_{j=a}^{a-1} c_j = 0$
and $\sum_{j=a}^b c_j = -\sum_{j=b+1}^{a-1} c_j$.

Given further $k\in\ZZ$, we denote the \emph{$q$-Gaussian falling factorial of $n$ by $k$}, the
\emph{$q$-Gaussian factorial} and the \emph{$q$-Gaussian binomial of $n$ by $k$} by
\begin{align*}
  (n)_{k,q} & \df \prod_{j=0}^{k-1} [n-j]_q,
  &
  k!_q & \df (k)_{k,q},
  &
  \binom{n}{k}_q & \df
  \begin{dcases*}
    \frac{(n)_{k,q}}{k!_q}, & if $k\geq 0$,\\
    0, & otherwise,
  \end{dcases*}
\end{align*}
respectively. When $k\leq 0$, products should be interpreted in the usual fashion so that
$\prod_{j=a}^{a-1} c_j = 1$ and $\prod_{j=a}^b c_j = \prod_{j=b+1}^{a-1} c_j^{-1}$. We will omit
$q$ from the notation when $q=1$, so that the above match the usual falling factorial, factorial
and binomial, respectively.

For a set $V$ and $k\in\ZZ$, we denote by $\binom{V}{k}$ the set of all subsets of $V$ of size $k$
(so $\lvert\binom{V}{k}\rvert = \binom{\lvert V\rvert}{k}$ when $V$ is finite).

For a prime power $q\in\NN$, we denote by $\FF_q$ the field with $q$ elements and for $x\in\FF_q^n$,
we denote by $\lvert x\rvert\df\lvert\supp(x)\rvert$ the \emph{Hamming weight} of $x$. For an
$\FF_q$-vector space $V$, we denote by $L_{\FF_q}(V)$ the set of all $\FF_q$-linear subspaces of $V$
and we denote by $\GL_\ell(\FF_q)$ the general linear group of degree $\ell$ over $\FF_q$ (i.e., the
group of non-singular $\ell\times\ell$ matrices over $\FF_q$). For a matrix $X$, we denote by $X_i$
the $i$th row of $X$ and by $X_{i_1,\ldots,i_t}$ the matrix obtained by restricting $X$ to the rows
indexed by $i_1,\ldots,i_t$.

A \emph{distance-$d$ code} is a code $C \subseteq \F_q^n$ such that $\lvert x-y\rvert\geq d$ for all
$x,y \in C$ with $x \neq y$.  We denote by $A_q(n,d)$ the size of the largest distance-$d$ code in
$\F_q^n$ and by $A_q^{\Lin}(n,d)$ the size of the largest distance-$d$ code in $\F_q^n$ that
is also a subspace of $\F_q^n$.

\end{document}